\documentclass[singlecolumn,aps,pra,superscriptaddress,floatfix,tightenlines,nofootinbib,nobibnotes]{revtex4-2}

\usepackage[colorlinks,linkcolor=blue,urlcolor=blue, citecolor=blue]{hyperref}
\usepackage{xurl}
\usepackage{booktabs}
\usepackage{graphicx}
\usepackage{comment}
\usepackage{nicematrix}
\usepackage{xcolor}

\usepackage{dsfont}
\usepackage{physics}
\usepackage{nicefrac}
\usepackage{mathtools}
\usepackage{amsmath, amssymb, amsfonts, amsthm}
\usepackage{thmtools, thm-restate}

\usepackage{mathrsfs}

\usepackage{algorithmic}
\usepackage{algorithm}

\usepackage{enumitem}

\newtheorem{theorem}{Theorem}
\newtheorem*{theorem*}{Theorem}

\newtheorem{corollary}[theorem]{Corollary}
\newtheorem{lemma}[theorem]{Lemma}
\newtheorem{definition}[theorem]{Definition}
\newtheorem{proposition}[theorem]{Proposition}
\newtheorem{example}[theorem]{Example}
\newtheorem{remark}[theorem]{Remark}

\DeclareMathOperator{\supp}{supp}

\allowdisplaybreaks

\makeatletter
\renewcommand{\@fnsymbol}[1]{%
  \ensuremath{%
    \ifcase#1
      \ast
    \or
      \dagger
    \or
      \ddagger
    \or
      \P
    \else
      \@ctrerr
    \fi
  }%
}
\makeatother

\begin{document}

\begin{abstract}
The generalized quantum Stein's lemma characterizes the optimal asymptotic exponent of the type-II error in quantum hypothesis testing for an independent and identically distributed (IID) null hypothesis against a composite alternative hypothesis.
Classically, a probabilistic mixture of IID sources arises as a natural generalization of IID sources, and, in the non-composite setting, the optimal type-II error exponent in hypothesis testing for such classical mixed sources is known to be characterized concisely by the worst-case component of the mixture.
In this work, we extend these foundational results to composite quantum hypothesis testing where the null hypothesis is a mixed source, i.e., a probabilistic mixture of IID quantum states, and the alternative hypothesis is composite as in the generalized quantum Stein's lemma.
When the type-I error vanishes asymptotically, we characterize the optimal type-II error exponent of this composite quantum hypothesis testing problem in terms of the worst-case component of the mixture, by developing techniques for the non-commutative quantum setting inspired by the classical information-spectrum analysis.
We also show that the analogous characterization does not hold in general for a fixed nonzero type-I error threshold, by providing a counterexample beyond the vanishing type-I error regime.
These results clarify the applicability of the generalized quantum Stein's lemma to highly non-IID null hypotheses arising from arbitrary finite probabilistic mixtures of IID quantum states.
\end{abstract}

\title{Generalized quantum Stein's lemma for mixed sources}

\author{Haruka Kanazawa$^{*,}$}
\email{kanazawa-haruka306@g.ecc.u-tokyo.ac.jp}
\affiliation{
Department of Information Science, School of Science, The University of Tokyo, 7-3-1 Hongo, Bunkyo-ku, Tokyo, 113-8656, Japan
}

\author{Hayata Yamasaki$^{*,}$}
\email{hayata.yamasaki@gmail.com}
\affiliation{
Department of Computer Science, Graduate School of Information Science and Technology, The University of Tokyo, 7-3-1 Hongo, Bunkyo-ku, Tokyo, 113-8656, Japan
}

\begingroup
\renewcommand{\thefootnote}{$*$}
\footnotetext{The authors contributed equally to this work.}
\endgroup
\setcounter{footnote}{0}

\maketitle

\tableofcontents

\section{Introduction}

Quantum hypothesis testing~\cite{hiai1991proper,ogawa2000strong} is a fundamental task in quantum information theory.
Given two possible descriptions of an unknown quantum state, the task is to perform a measurement on the quantum state to distinguish the null hypothesis from the alternative hypothesis with small error probabilities.
The two errors are asymmetric: the type-I error is the probability of rejecting the null hypothesis when it is true, while the type-II error is the probability of accepting the null hypothesis when the alternative hypothesis is true.
A central question is to determine the optimal asymptotic decay rate of the type-II error under a constraint on the type-I error.

In the classical IID setting, this rate is characterized by the Chernoff-Stein lemma~\cite{cover1999elements}, which states that the optimal asymptotic type-II error exponent is given by the relative entropy between the two underlying probability distributions.
Its quantum counterpart, known as the quantum Stein's lemma, was established by Hiai and Petz~\cite{hiai1991proper} and by Ogawa and Nagaoka~\cite{ogawa2000strong}, showing that the same operational role is played by Umegaki's quantum relative entropy~\cite{umegaki1962conditional}.
Beyond the IID setting, the classical information-spectrum method provides a general framework for analyzing hypothesis testing of arbitrary non-IID sequences of probability distributions~\cite{han2003,Han2000HypothesisTestingGeneralSource}.
A quantum extension of this framework was developed by Hayashi and Nagaoka~\cite{4069150}, where the optimal type-II error exponent for general sequences of quantum states is characterized by the spectral inf-divergence rate.
These results form a general foundation for hypothesis testing beyond independently and identically distributed sources.
Together with various other asymptotic formulations of quantum state discrimination tasks, including Chernoff bounds~\cite{chernoff1952measure,Audenaert_2007,10.1214/08-AOS593} and Hoeffding bounds~\cite{hoeffding1965asymptotically,ogawa2002errorexponentsquantumhypothesis,PhysRevA.76.062301,nagaoka2006conversetheoremquantumhoeffding,audenaert2008asymptotic}, they constitute a foundation of quantum information theory.

A particularly important extension of quantum hypothesis testing is the composite hypothesis testing in which the alternative hypothesis is not a single state but a set of possible states.
This setting appears naturally in quantum resource theories, where the alternative hypothesis may be a family of free states satisfying physically well-motivated axioms, such as the set of separable states in entanglement theory~\cite{Gour_2025,Chitambar_2019}.
Under such axioms, the generalized quantum Stein's lemma characterizes the optimal type-II error exponent for an IID null hypothesis against a composite alternative hypothesis in terms of the regularized relative entropy of resources~\cite{Brand_o_2008,brandao2010reversible,brandao2010generalization,Brandao2015,Berta_2023,hayashi2025generalizedquantumsteinslemma,10898013}.
The original formulation by Brand\~{a}o and Plenio~\cite{brandao2010generalization} proposed axioms including closedness under taking partial traces and invariance under permutations of subsystems.
Reference~\cite{hayashi2025generalizedquantumsteinslemma} showed that the generalized quantum Stein's lemma holds under a smaller set of axioms on the sequence of alternative-hypothesis sets, such as the existence of a full-rank state, compactness, closedness under tensor products, and convexity, while Ref.~\cite{10898013} also showed the generalized quantum Stein's lemma under the original Brand\~{a}o-Plenio axioms.
More recently, Refs.~\cite{hayashi2025generalizedquantumsteinslemmaCQ,bergh2025generalizedquantumsteinslemma} generalized these results from states to classical-quantum channels.
Note that quantum hypothesis testing settings with composite alternative hypotheses have also been studied in Refs.~\cite{fang2025generalizedquantumasymptoticequipartition,fang2025errorexponentsquantumstate,lami2025doublycompositechernoffsteinlemma}, while the settings of these works are incompatible with the application domains originally proposed by Brand\~{a}o and Plenio.
These developments clarify the operational meaning of the regularized relative entropy of resources and provide a broad foundation for quantum resource theories.

The existing analyses of the generalized quantum Stein's lemma, however, assume that the null hypothesis is IID, while information theory for general source coding and the theory of information-spectrum methods have further developed to deal with correlated non-IID sources.
Reference~\cite{10898013} showed that its proof strategy for the generalized quantum Stein's lemma also applies when the null hypothesis is an almost IID state rather than an IID state, while extension of the result to highly non-IID null hypotheses beyond that setting inevitably requires different techniques.
In classical information theory, mixed sources, i.e., probabilistic mixtures of IID sources, arise as a natural and fundamental class of non-IID sources~\cite{han2003}.
Mixed sources model situations in which a source is chosen randomly from several IID components and then kept fixed throughout.
Such sources are generally non-ergodic and are not reducible to a single IID distribution.
In the non-composite classical hypothesis-testing setting, the optimal type-II error exponent for mixed sources admits a concise characterization: when the type-I error is required to vanish asymptotically, the exponent is determined by the worst component of the mixture~\cite{Han2000HypothesisTestingGeneralSource,han2003}.
This characterization is one of the basic examples illustrating how information-spectrum methods capture non-IID behavior.

In this work, we extend these foundational results on the generalized quantum Stein's lemma and the classical mixed-source Stein's lemma to a unified setting of composite quantum hypothesis testing.
We consider a null hypothesis given by a mixed source
$\rho_n=\sum_{i\in J}p_i\overline{\rho}_i^{\otimes n}$,
where $J$ is finite, $\{p_i\}_{i\in J}$ is an arbitrary probability distribution, and each $\overline{\rho}_i$ is a density operator on a finite-dimensional Hilbert space.
The alternative hypothesis is a sequence $\vec{S}=\{S_n\}_{n=1}^{\infty}$ of compact convex sets of quantum states satisfying the assumptions used in the proof of the generalized quantum Stein's lemma in Refs.~\cite{hayashi2025generalizedquantumsteinslemma,hayashi2025generalizedquantumsteinslemmaCQ}.

Our main theorem shows that, when the type-I error vanishes asymptotically, the optimal type-II error exponent is given by the worst IID component in the mixture, establishing an analogue of classical hypothesis testing for mixed sources~\cite{han2003} in the composite quantum setting with a non-IID mixed-source null hypothesis and a composite alternative hypothesis.
In the classical information-spectrum method, the corresponding mixed-source result follows from properties of the probabilistic spectral quantity $\mathrm{p}$-$\liminf$; however, in the quantum setting, the relevant spectral projections are non-commutative, and hence the classical argument does not directly apply.
We overcome this difficulty by proving that the spectral inf-divergence rate of a finite quantum mixed source is equal to the minimum of the spectral inf-divergence rates of its IID components.
The direct part is the technically nontrivial direction.
We construct a single test from the componentwise spectral tests by thresholding the sum of the tests, and control the type-I error for each component using an operator inequality derived from the Cauchy-Schwarz inequality.
This quantum extension applies without any additional axiom beyond those imposed on the composite alternative hypothesis in Ref.~\cite{hayashi2025generalizedquantumsteinslemma}.

We also analyze the extent to which the worst-component characterization remains valid when the type-I error threshold is fixed at a nonzero value.
For IID null hypotheses, the generalized quantum Stein's lemma has a strong-converse form: the exponent is independent of any fixed type-I error threshold $\varepsilon\in(0,1)$; by contrast, for mixed-source null hypotheses, we show that this behavior fails in general.
Even when the alternative hypothesis is a single IID state, the optimal exponent for a mixed source depends on how the allowed type-I error can discard low-relative-entropy components of the mixture.
This gives a counterexample to the direct extension of the worst-component characterization to fixed nonzero type-I error thresholds, and the exponent can jump discontinuously as the probability distribution defining the mixture is varied.
The proof of this nonzero-error characterization uses a complementary method distinct from the above result for vanishing type-I errors.
We show that, when the number of distinct eigenvalues of the alternative-hypothesis sequence grows subexponentially, the pinching map~\cite{Hayashi_2002} preserves the spectral inf-divergence rate up to asymptotically vanishing terms.
This allows us to reduce the problem to a commuting, hence classical, setting and then apply the classical information-spectrum formula for mixed sources~\cite{han2003}.

Therefore, these results clarify both the scope and the limitation of the generalized quantum Stein's lemma for highly non-IID null hypotheses arising from arbitrary finite mixtures of IID quantum states.
In the vanishing type-I error regime, the generalized quantum Stein's lemma extends to null hypotheses given by arbitrary finite probabilistic mixtures of IID quantum states, with the exponent determined by the least distinguishable component.
However, the strong-converse behavior in the setting of IID null hypotheses does not survive when the null hypothesis is a mixed source.
This distinction highlights a genuinely non-IID feature of mixed sources and shows that the probability distribution of the mixture, although irrelevant in the vanishing type-I error regime, becomes operationally relevant under a fixed nonzero type-I error constraint.

The rest of this paper is organized as follows.
In Sec.~\ref{sec:formulation}, we introduce the setting of composite hypothesis testing in this work, along with relevant quantities from the information-spectrum method.
In Sec.~\ref{sec:generalized_qsl_mixed_sources}, we prove the generalized quantum Stein's lemma for mixed sources.
In Sec.~\ref{sec:counterexample_epsilon}, we analyze the setting with fixed nonzero type-I error thresholds and provide the counterexample.
Finally, in Sec.~\ref{sec:conclusion}, we provide the conclusion and outlook.

\section{Formulation}
\label{sec:formulation}

Let $\mathcal{D}(\mathcal{H})$ denote the set of density operators on a finite-dimensional Hilbert space $\mathcal{H}$.
Quantum hypothesis testing is formalized as follows:
we are given a sequence $\vec{\mathcal{H}}=\left\{\mathcal{H}_n\right\}_{n=1}^{\infty}$ of Hilbert spaces and 
two sequences $\vec{\rho}=\left\{\rho_n\right\}_{n=1}^{\infty}$ and $\vec{\sigma}=\left\{\sigma_n\right\}_{n=1}^{\infty}$ of quantum states, 
where $\rho_n\in \mathcal{D}(\mathcal{H}_n)$ and $\sigma_n\in \mathcal{D}(\mathcal{H}_n)$ for each $n$.
Let $I_n$ denote the identity operator on $\mathcal{H}_n$ for each $n$.
Our task is to distinguish between $\vec{\rho}$ and $\vec{\sigma}$ by performing measurements represented by a sequence $\vec{T}=\left\{T_n\right\}_{n=1}^{\infty}$ of elements of positive operator-valued measures (POVMs) $\{T_n,I_n-T_n\}$, where each $T_n$ satisfies $0\le T_n \le I_n$ and is defined on $\mathcal{H}_n$.
In this setting, the null hypothesis is $\vec{\rho}$, and the alternative hypothesis is $\vec{\sigma}$.
If the measurement outcome corresponding to $T_n$ is obtained, we conclude that the given state
is $\rho_n$; otherwise, we conclude that the given state is $\sigma_n$.
We define two types of errors:
\begin{itemize}
    \item Type-I error: A mistaken conclusion that we reject the null hypothesis $\vec{\rho}$ when it is true, i.e., 
    we obtain the measurement outcome corresponding to $I_n - T_n$ when the given state is $\rho_n$. 
    The probability of this error is given by $\alpha_n(T_n) \coloneqq \Tr[(1 - T_n)\rho_n]$.
    \item Type-II error: A mistaken conclusion that we accept the null hypothesis $\vec{\rho}$ when the alternative hypothesis $\vec{\sigma}$ is true, i.e., 
    we obtain the measurement outcome corresponding to $T_n$ when the given state is $\sigma_n$.
    The probability of this error is given by $\beta_n(T_n) \coloneqq \Tr[T_n \sigma_n]$.
\end{itemize}
Classically, the hypothesis testing for such general sequences of sources is studied using techniques from the information spectrum method~\cite{han2003,Han2000HypothesisTestingGeneralSource}.
In the current quantum setting, we may use its quantum extension formulated in Ref.~\cite{4069150}.
Given a target error $\varepsilon\in[0,1]$, we aim to minimize the type-II errors while restricting the type-I errors below $\varepsilon$; as $n$ increases, the type-II errors decay exponentially, and this decay exponent is characterized by the following quantity.
\begin{definition}[Supremum type-II error exponents in $\varepsilon$-hypothesis testing for general sequences of quantum states~\cite{4069150}]
\label{def:B}
    For two sequences of states $\vec{\rho}=\left\{\rho_n\right\}_{n=1}^{\infty}$, $\vec{\sigma}=\left\{\sigma_n\right\}_{n=1}^{\infty}$ defined on a sequence of Hilbert spaces $\vec{\mathcal{H}}=\left\{\mathcal{H}_n\right\}_{n=1}^{\infty}$, 
    and any $\varepsilon\in[0,1]$,
    we define the supremum $\varepsilon$-achievable type-II error exponent as
    \begin{align}
        B\left(
            \varepsilon | \vec{\rho} \| \vec{\sigma}
        \right)
        &\coloneqq
        \sup_{\vec{T}=\{T_n\}_{n=1}^{\infty}}\left\{
            \liminf_{n\to\infty} -\frac{1}{n} \log \beta_n(T_n) \middle| \limsup_{n\to\infty} \alpha_n(T_n) \le \varepsilon
        \right\},\label{B}
\end{align}
where $\log$ is the natural logarithm throughout our work.
In particular, when $\varepsilon=0$, we define the supremum type-II error exponent as
\begin{align}
    B(\vec{\rho}\|\vec{\sigma}) \coloneqq B(0|\vec{\rho}\|\vec{\sigma}).
\end{align}
\end{definition}

To analyze this operational quantity, Ref.~\cite{4069150} introduces the quantum extension of the spectral inf-divergence rate, based on the classical information spectrum method~\cite{han2003}.
\begin{definition}[The spectral inf-divergence rate in $\varepsilon$-hypothesis testing for general sequences of quantum states]\cite{4069150}\label{def D}
    For two sequences of states $\vec{\rho}=\left\{\rho_n\right\}_{n=1}^{\infty}$, $\vec{\sigma}=\left\{\sigma_n\right\}_{n=1}^{\infty}$ defined on a sequence of Hilbert spaces $\vec{\mathcal{H}}=\left\{\mathcal{H}_n\right\}_{n=1}^{\infty}$, 
    and any $\varepsilon\in[0,1]$, we define
    \begin{align}
        \underline{D}\left(
            \varepsilon | \vec{\rho} \| \vec{\sigma}
        \right)
        &\coloneqq
        \sup\left\{
            a \:\middle|\: \limsup_{n\to\infty} \Tr\left[
                \left(I_n - \left\{ \rho_n - e^{na}\sigma_n \ge 0 \right\} \right)\rho_n
            \right] \le \varepsilon
        \right\},
    \end{align}
    where $\{X\ge 0\}$ denotes the projection onto the eigenspace of $X$ corresponding to non-negative eigenvalues.
    In particular, we define $B(\vec{\rho}\|\vec{\sigma}) \coloneqq B(0|\vec{\rho}\|\vec{\sigma})$ and $\underline{D}(\vec{\rho}\|\vec{\sigma}) \coloneqq \underline{D}(0|\vec{\rho}\|\vec{\sigma})$.
\end{definition}

Importantly, Ref.~\cite[Theorem 1]{4069150} shows the following fact connecting these quantities.
\begin{lemma}[Characterization of the supremum type-II error exponent]
\label{lem:B=D}
For any $\varepsilon\in[0,1]$, it holds that
\begin{align}\label{B=D}
    B\left(
        \varepsilon | \vec{\rho} \| \vec{\sigma}
    \right)
=
\underline{D}\left(
    \varepsilon | \vec{\rho} \| \vec{\sigma}
\right).
\end{align}
\end{lemma}

We extend the setting of quantum hypothesis testing to composite hypothesis testing, where the alternative hypothesis is not a single sequence of states but a sequence of sets of states.
In this case, we are given a sequence of sets of states $\vec{S}=\left\{S_n\right\}_{n=1}^{\infty}$, where $S_n \subseteq \mathcal{D}(\mathcal{H}_n)$, and the alternative hypothesis is $\vec{S}$.
Following the previous works~\cite{hayashi2025generalizedquantumsteinslemma, hayashi2025generalizedquantumsteinslemmaCQ}, we assume that $\vec{S}$ satisfies the following conditions:
\begin{enumerate}[label=(Axiom \arabic*), ref=Axiom \arabic*]
    \item \label{set1}
    The set $S_1$ contains a full-rank state $\sigma_{\mathrm{full}}$.
    \item \label{set2} 
    For a $d$-dimensional Hilbert space $\mathcal{H}$ ($0<d<\infty$), for each $n$, the set $S_n$ is a compact subset of $\mathcal{D}(\mathcal{H}^{\otimes n})$.
    \item \label{set3}
    The set $S_n$ is closed under tensor product; i.e., for any positive integers $n$ and $m$,
    if $\sigma_n \in S_n$ and $\sigma_m \in S_m$, then $\sigma_n \otimes \sigma_m \in S_{n+m}$.
    \item \label{set4} 
    For a $d$-dimensional Hilbert space $\mathcal{H}$ ($0<d<\infty$), for each $n$, the set $S_n$ is a convex subset of $\mathcal{D}(\mathcal{H}^{\otimes n})$.
\end{enumerate}

To analyze this composite hypothesis testing, we recall the definition of the optimal type-II error used in Ref.~\cite{hayashi2025generalizedquantumsteinslemma}.
We will analyze the relationship between the two definitions (Definitions~\ref{def:B} and~\ref{def beta}) in Sec.~\ref{subsec: the relationship between two definitions of type-II error exponent}.
\begin{definition}[Optimal type-II error in the generalized quantum Stein's lemma]
\label{def beta}
    For two density operators $\rho$ and $\sigma$ defined on the same Hilbert space, and $\varepsilon \in [0, 1]$, we define the optimal $\varepsilon$-achievable type-II error as
    \begin{align}
        \beta_\varepsilon(\rho\|\sigma)  \coloneqq \min \left\{
            \Tr[T\sigma] \middle| 0 \le T \le I, \Tr[(I-T)\rho] \le \varepsilon
        \right\}.
    \end{align}
    More generally, for a state $\rho$ and a compact set  $S$ of states, define
    \begin{align}
        \beta_\varepsilon(\rho\|S) \coloneqq
        \min \left\{
            \max_{\sigma\in S} \Tr[T\sigma] \middle| 0 \le T \le I, \Tr[(I-T)\rho] \le \varepsilon
        \right\}.
    \end{align}
\end{definition}

The quantum relative entropy is defined as follows.
\begin{definition}[Quantum relative entropy]
For two density operators $\rho$ and $\sigma$ defined on the same Hilbert space, the quantum relative entropy is defined by
\begin{align}
    D(\rho\|\sigma)
    =
    \Tr[\rho(\log \rho - \log \sigma)]
\end{align}
if $\supp(\rho)\subseteq \supp(\sigma)$, and $D(\rho\|\sigma)=+\infty$ otherwise, 
where $\supp(\rho)$ denotes the support of $\rho$.
\end{definition}

The following theorem is known as the generalized quantum Stein's lemma~\cite[Theorem 4]{hayashi2025generalizedquantumsteinslemma}.
We note that the original proposal of the generalized quantum Stein's lemma in Ref.~\cite{brandao2010generalization} considered additional conditions on the closedness under taking partial traces and the invariance under permutations of subsystems, for which Ref.~\cite{10898013} also provides an alternative proof to the proofs in Refs.~\cite{hayashi2025generalizedquantumsteinslemma, hayashi2025generalizedquantumsteinslemmaCQ}.
Here, we work on the smaller set of conditions established in Ref.~\cite{hayashi2025generalizedquantumsteinslemma, hayashi2025generalizedquantumsteinslemmaCQ} to make the statements stronger. The right-hand side of \eqref{eq: generalized quantum Stein's lemma} is called the regularized relative entropy of resources~\cite{Gour_2025,Chitambar_2019,Kuroiwa_2020}.

\begin{theorem}[{Generalized quantum Stein's lemma~\cite[Theorem 4]{hayashi2025generalizedquantumsteinslemma}}]\label{generalized quantum Stein's lemma}
    Let $\vec{\rho}=\left\{ \rho^{\otimes n} \right\}_{n=1}^{\infty}$ be a sequence of density operators defined on a sequence $\vec{\mathcal{H}}=\left\{\mathcal{H}_n\right\}_{n=1}^{\infty}$ of Hilbert spaces.
    For any $\varepsilon \in (0, 1)$ and any sequence of sets of states $\vec{S}=\left\{S_n\subset\mathcal{D}(\mathcal{H}_n)\right\}_{n=1}^{\infty}$ satisfying \textnormal{(\ref{set1})}, \textnormal{(\ref{set2})}, \textnormal{(\ref{set3})}, and \textnormal{(\ref{set4})},
    it holds that
    \begin{align}\label{eq: generalized quantum Stein's lemma}
        \lim_{n\to\infty} -\frac{1}{n} \log \beta_\varepsilon(\rho^{\otimes n}\| S_n)
        =
        \lim_{n\to\infty} \frac{1}{n} \min_{\sigma_n \in S_n} D(\rho^{\otimes n}\|\sigma_n).
    \end{align}
\end{theorem}

\section{Generalized quantum Stein's lemma for mixed sources}
\label{sec:generalized_qsl_mixed_sources}

In this section, we present our main result: a generalization of the generalized quantum Stein's lemma~\cite{brandao2010generalization, hayashi2025generalizedquantumsteinslemma, hayashi2025generalizedquantumsteinslemmaCQ} to mixed sources defined as a probabilistic mixture of independent and identically distributed (IID) quantum states.
In Sec.~\ref{subsec: generalized quantum Stein's lemma for mixed sources}, we present the main theorem in this work. The proof of the main theorem is divided into four subsections. Section~\ref{subsec: analysis of the spectral inf-divergence rate} proves the key proposition to prove the main theorem. Section~\ref{subsec: the spectral inf-divergence rate in composite hypothesis testing} shows a corollary of the proposition proved in Sec.~\ref{subsec: analysis of the spectral inf-divergence rate}. In Sec.~\ref{subsec: the relationship between two definitions of type-II error exponent}, we prove a relationship between two different ways to define the optimal type-II error exponent so that we can relate Proposition~\ref{key proposition to the main theorem} to the known result in Ref.~\cite{hayashi2025generalizedquantumsteinslemma}. Finally, the proof of the main theorem is provided in Sec.~\ref{subsec: proof of the main theorem}.

\subsection{Statement of the result}
\label{subsec: generalized quantum Stein's lemma for mixed sources}

We present the main theorem, Theorem~\ref{Main theorem}, on the generalized quantum Stein's lemma for mixed sources.
The key to proving Theorem~\ref{Main theorem} is  Proposition~\ref{key proposition to the main theorem} shown in Sec.~\ref{subsec: the spectral inf-divergence rate in composite hypothesis testing}, indicating that the spectral inf-divergence rate for mixed sources of quantum states is determined by the worst case in the mixture.
The corresponding result has long been known in the classical case~\cite{Han2000HypothesisTestingGeneralSource, han2003}.
However, it is non-trivial to show Proposition~\ref{key proposition to the main theorem} because the proof of the classical counterpart of Proposition~\ref{key proposition to the main theorem} utilizes characteristics of $\mathrm{p}$-$\liminf$ in the information spectrum method~\cite{4069150}, which appears in the definition of $\underline{D}(\varepsilon|\vec{\rho}\|\vec{\sigma})$ only in the classical case. Moreover, it is generally not easy to derive some relations between different projections $\{X\ge 0\}$, which
appear in Definition~\ref{def D}.
Nevertheless, we prove Proposition~\ref{key proposition to the main theorem} and thus Theorem~\ref{Main theorem} 
by using a fundamentally different technique from the classical case, just as Ref.~\cite{4069150} used a completely different method to prove quantum Stein's lemma. In particular, in the proof, we employ an operator inequality derived from the Cauchy-Schwarz inequality (Lemma~\ref{lem:Cauchy_Schwarz}) and use it to explicitly construct an optimal test operator from the test operators for each IID source.

\begin{theorem}[Generalized quantum Stein's lemma for mixed sources]\label{Main theorem}
    For each $n$, let $S_n$ be a subset of $\mathcal{D}(\mathcal{H}_n)$ satisfying \textnormal{(\ref{set1})}, \textnormal{(\ref{set2})}, \textnormal{(\ref{set3})}, and \textnormal{(\ref{set4})},
    and $\vec{S}$ be the sequence $\left\{S_n\right\}_{n=1}^{\infty}$.
    Suppose that $\vec{\rho}=\left\{\rho_n\right\}_{n=1}^{\infty}$ is a sequence of mixed sources of $\vec{\rho}_i=\left\{\overline{\rho}_i^{\otimes n}\right\}_{n=1}^{\infty}$ defined on $\vec{\mathcal{H}}$; i.e.,
    \begin{align}
    \label{eq:mixed_source_theorem}
        \rho_n = \sum_{i\in J} p_i \overline{\rho}_i^{\otimes n},
    \end{align}
    where $J$ is a finite set, $p_i>0$ for all $i\in J$, and $\sum_{i\in J} p_i = 1$.
    Then, it holds that
    \begin{align}
        \inf_{\vec{\sigma}\in\vec{S}} B(\vec{\rho}\|\vec{\sigma}) = \min_{i\in J} \lim_{n\to\infty} \frac{1}{n} \min_{\sigma_n \in S_n} D(\overline{\rho}_i^{\otimes n}\|\sigma_n).
    \end{align}
\end{theorem}

\begin{remark}[Extension to the sequence of sets with growing sizes]
    Our proof of this theorem holds even if we consider a corresponding statement with the finite set $J$ replaced with a sequence $\{J_n\}_{n=1}^\infty$ of potentially different sets with subexponential sizes $|J_n|=\exp[o(n)]$.
    In the information spectrum methods, on which our proof is based, the consistency condition between the elements in the sequences is unnecessary~\cite{han2003}.
    However, as shown by the following counterexample, when the size of these sets grows exponentially $|J_n|=\exp[\Theta(n)]$, the theorem no longer holds.
\end{remark}

\begin{example}[A counterexample when $|J|$ is exponentially large in $n$]\label{counterexample if J_n grows}
We show that the claimed lower bound can fail even in the classical case
when the number of mixture components is exponential in $n$.

Fix constants $d>R>0$. For each $n$, let
\begin{align}
J_n=\{1,\ldots,M_n\},
\qquad
M_n=e^{nR}.
\end{align}
For each $i\in J_n$, define the component distribution
\begin{align}
\rho_{i,n}\coloneqq \delta_i,
\end{align}
where $\delta_i$ is the point mass at $i$, i.e., $\rho_{i,n}(i)=1$, and $\rho_{i,n}(i')=0$ for every $i'\neq i$. Define the reference distribution
$\sigma_n$ by
\begin{align}
\sigma_n(i)=e^{-nd}
\qquad (i\in J_n),
\end{align}
and put the remaining probability mass
\begin{align}
1-\sum_{i\in J_n}\sigma_n(i)
=
1-M_ne^{-nd}
=
1-e^{-n(d-R)}
\end{align}
on one additional point outside $J_n$. Since $d>R$, this indeed defines
a probability distribution for all sufficiently large $n$.

For each component $i\in J_n$, we have
\begin{align}
\frac{\rho_{i,n}(i)}{\sigma_n(i)}
=
e^{nd}.
\end{align}
Hence, for each $i\in J_n$, we have
\begin{align}
\underline{D}(\vec{\rho}_i\|\vec{\sigma})=d,
\end{align}
and therefore
\begin{align}
\min_{i\in J_n}\underline{D}(\vec{\rho}_i\|\vec{\sigma})=d.
\end{align}

Now consider the uniform mixture
\begin{align}
\rho_n
:=
\frac1{M_n}\sum_{i=1}^{M_n}\rho_{i,n}
=\frac1{M_n}\sum_{i=1}^{M_n}\delta_i.
\end{align}
Then, for every $i\in J_n$,
\begin{align}
\rho_n(i)=\frac1{M_n}=e^{-nR}.
\end{align}
Thus,
\begin{align}
\frac{\rho_n(i)}{\sigma_n(i)}
=
\frac{e^{-nR}}{e^{-nd}}
=
e^{n(d-R)}.
\end{align}
It follows that
\begin{align}
\underline{D}(\vec{\rho}\|\vec{\sigma})=d-R.
\end{align}
Consequently,
\begin{align}
\underline{D}(\vec{\rho}\|\vec{\sigma})
=
d-R
<
d
=
\min_{i\in J_n}\underline{D}(\vec{\rho}_i\|\vec{\sigma}).
\end{align}
Therefore, for an exponentially large mixture, the bound
\begin{align}
\underline{D}(\vec{\rho}\|\vec{\sigma})
\ge
\min_i \underline{D}(\vec{\rho}_i\|\vec{\sigma})
\end{align}
does not hold in general, even classically.
\end{example}

\subsection{Analysis of the spectral inf-divergence rate}
\label{subsec: analysis of the spectral inf-divergence rate}

In this section, we prove the following proposition, which is key to the proof of Theorem \ref{Main theorem}.
The proposition shows that the spectral inf-divergence rate for mixed sources of quantum states is determined by the worst case in the mixture.

\begin{proposition}[Distribution of the spectral inf-divergence rate over mixed sources]\label{key proposition to the main theorem}
    Suppose that $\vec{\rho}$ is a sequence of mixed sources of $\vec{\rho}_i=\left\{\overline{\rho}_i^{\otimes n}\right\}_{n=1}^{\infty}$ defined on $\vec{\mathcal{H}}$ such that
    \begin{align}
        \rho_n = \sum_{i\in J} p_i \overline{\rho}_i^{\otimes n},
    \end{align}
    where $J$ is a finite set, $p_i>0$ for all $i\in J$, and $\sum_{i\in J} p_i = 1$.
    Let $\vec{\sigma}=\left\{\sigma_n\right\}_{n=1}^{\infty}$ be a sequence of density operators defined on $\vec{\mathcal{H}}$.
    Then, it holds that
    \begin{align}
        \underline{D}\left( \vec{\rho} \| \vec{\sigma} \right)
        =\label{main proposition eq}
        \min_{i\in J}
        \left\{
            \underline{D}\left( \vec{\rho}_i \| \vec{\sigma} \right)
        \right\}.
    \end{align}
\end{proposition}
\begin{proof}
We divide the proof of this proposition into the two inequalities.
Proposition~\ref{main proposition: direct} shows the direct part
\begin{align}
    \underline{D}\left( \vec{\rho} \| \vec{\sigma} \right)
    \ge\label{main proposition: direct eq}
    \min_{i\in J}
    \left\{
        \underline{D}\left( \vec{\rho}_i \| \vec{\sigma} \right)
    \right\}.
\end{align}
Proposition~\ref{main proposition: converse} shows the converse part
\begin{align}
    \underline{D}\left( \vec{\rho} \| \vec{\sigma} \right)
    \le\label{main proposition: converse eq}
    \min_{i\in J}
    \left\{
        \underline{D}\left( \vec{\rho}_i \| \vec{\sigma} \right)
    \right\}.
\end{align}
As a whole,~ \eqref{main proposition: direct eq} and~\eqref{main proposition: converse eq} imply \eqref{main proposition eq}.
\end{proof}

The proof of this proposition necessitates the following lemma.
It is generally difficult to establish a tight operator inequality between projections, but this lemma shows a useful relation in terms of their traces.

\begin{lemma}[Operator inequality from Cauchy-Schwarz]
\label{lem:Cauchy_Schwarz}
Let $T$ be any operator satisfying $0\le T\le I$, $\Pi$ be any projection, and $c\in[0,1]$ be any real number satisfying
\begin{align}
\Pi T\Pi\le c\Pi.
\end{align}
Then, for any state $\rho$, we have
\begin{align}
\Tr[T\rho]
\le
\left(\sqrt{c\,\Tr[\Pi \rho]}+\sqrt{1-\Tr[\Pi \rho]}\right)^2.
\end{align}
\end{lemma}

\begin{proof}
With $\Pi'\coloneqq I-\Pi$, we decompose
\begin{align}
\Tr[T\rho]
&=
\Tr[\Pi T\Pi\rho]
+
\Tr[\Pi' T\Pi'\rho]
+
\Tr[\Pi T\Pi'\rho]
+
\Tr[\Pi' T\Pi\rho].
\end{align}
The first term is bounded as
\begin{align}
\Tr[\Pi T\Pi\rho]\le c\,\Tr[\Pi\rho],
\end{align}
because $\Pi T\Pi\le c\Pi$.
The second term is bounded as
\begin{align}
\Tr[\Pi' T\Pi'\rho]\le \Tr[\Pi' \rho]=1-\Tr[\Pi\rho],
\end{align}
because $0\le T\le I$.
For the off-diagonal terms, the Cauchy-Schwarz inequality gives
\begin{align}
\left|\Tr[\Pi T\Pi'\rho]\right|=\left|\Tr[\sqrt{\rho}\Pi\sqrt{T}\,\sqrt{T}\Pi'\sqrt{\rho}]\right|
\le
\sqrt{
\Tr[\Pi T\Pi\rho]\,
\Tr[\Pi' T\Pi'\rho]
}.
\end{align}
The same bound holds for $\Tr[\Pi' T\Pi\rho]$.
Therefore, we obtain
\begin{align}
\Tr[T\rho]
&\le
c\,\Tr[\Pi\rho]
+
\left(1-\Tr[\Pi\rho]\right)
+
2\sqrt{
c\,\Tr[\Pi\rho]
\left(1-\Tr[\Pi\rho]\right)
} \\
&=
\left(\sqrt{c\,\Tr[\Pi\rho]}+\sqrt{1-\Tr[\Pi\rho]}\right)^2.
\end{align}
\end{proof}

Using this lemma, the direct part of Proposition~\ref{key proposition to the main theorem} is shown as follows.

\begin{proposition}[Direct part]\label{main proposition: direct}
    Suppose that $\vec{\rho}$ is a sequence of mixed sources of $\vec{\rho}_i=\left\{\overline{\rho}_i^{\otimes n}\right\}_{n=1}^{\infty}$ defined on $\vec{\mathcal{H}}$ such that
    \begin{align}
        \rho_n = \sum_{i\in J} p_i \overline{\rho}_i^{\otimes n},
    \end{align}
    where $J$ is a finite set, $p_i>0$ for all $i\in J$, and $\sum_{i\in J} p_i = 1$.
    Let $\vec{\sigma}=\left\{\sigma_n\right\}_{n=1}^{\infty}$ be a sequence of density operators defined on $\vec{\mathcal{H}}$.
    Then, it holds that
    \begin{align}
        \underline{D}\left( \vec{\rho} \| \vec{\sigma} \right)
        \ge
        \min_{i\in J}
        \left\{
            \underline{D}\left( \vec{\rho}_i \| \vec{\sigma} \right)
        \right\}.
    \end{align}
\end{proposition}
\begin{proof}
Fix arbitrary $a$ satisfying
\begin{align}
\label{eq:a_condition}
    a<\min_i\left\{\underline{D}(\vec{\rho}_i\|\vec{\sigma})\right\}.
\end{align}
It suffices to show that
\begin{align}
\underline{D}(\vec{\rho}\|\vec{\sigma})\geq a.
\end{align}

Choose $\delta>0$ such that
\begin{align}
a+2\delta<\min_i\left\{\underline{D}(\vec{\rho}_i\|\vec{\sigma})\right\}.
\end{align}
Since $a+2\delta<\underline{D}(\vec{\rho}_i\|\vec{\sigma})$ for every $i$, by the definition of
$\underline{D}(\vec{\rho}_i\|\vec{\sigma})$, the test operator
\begin{align}
T_{i,n}\coloneqq
\{\overline{\rho}_i^{\otimes n}-e^{n(a+2\delta)}\sigma_n\geq 0\}
\end{align}
satisfies
\begin{align}
    \label{eq:T_n_convergence}
\lim_{n\to\infty}\Tr\left[(I-T_{i,n})\overline{\rho}_i^{\otimes n}\right]= 0.
\end{align}
Moreover, due to
\begin{align}
\Tr[T_{i,n}(\overline{\rho}_i^{\otimes n}-e^{n(a+2\delta)}\sigma_n)]\geq 0,
\end{align}
we have
\begin{align}
\Tr\!\left[T_{i,n}\sigma_n\right]
\le
e^{-n(a+2\delta)}
\Tr\!\left[T_{i,n}\overline{\rho}_i^{\otimes n}\right]
\le
e^{-n(a+2\delta)}.
\end{align}

We now combine all the tests.
Define
\begin{align}
A_n\coloneqq \sum_{i\in J} T_{i,n},
\qquad
T_n\coloneqq \{A_n\ge e^{-n\delta}I\}.
\end{align}
Then, it holds that $0\le T_n\le I$, so $T_n$ is a valid test operator.
For the spectral decomposition
\begin{align}
    A_n=\sum_j\lambda_j\Pi_j,
\end{align}
we have
\begin{align}
T_n=\sum_{j:\lambda_j\geq e^{-n\delta}}\Pi_j
\le e^{n\delta}A_n.
\end{align}
Thus, we obtain
\begin{align}
\Tr\!\left[T_n\sigma_n\right]
&\le
e^{n\delta}\Tr\!\left[A_n\sigma_n\right] \\
&=
e^{n\delta}
\left(
\sum_{i\in J}
\Tr\!\left[T_{i,n}\sigma_n\right]
\right) \\
&\le
e^{n\delta}
\left(
\sum_{i\in J}
e^{-n(a+2\delta)}
\right)\\
&=
|J|e^{-n(a+\delta)}.
\end{align}
Thus
\begin{align}
        \label{eq:T_n_II}
\liminf_{n\to\infty}
-\frac1n\log \Tr\!\left[T_n\sigma_n\right]
\ge (a+\delta)-\liminf_{n\to\infty}\frac{\log|J|}{n}=a+\delta.
\end{align}

It remains to show that $T_n$ accepts each component state with probability tending to one.
Let
\begin{align}
Q_n:=I_n-T_n=\{A_n<e^{-n\delta}I\}.
\end{align}
Since $Q_n$ is the spectral projection of $A_n$ corresponding to eigenvalues smaller than
$e^{-n\delta}$, we have
\begin{align}
Q_n A_n Q_n\le e^{-n\delta}Q_n.
\end{align}
Moreover, since $0\le T_{i,n}\le A_n$, it follows that
\begin{align}
Q_nT_{i,n}Q_n\le Q_nA_nQ_n\le e^{-n\delta}Q_n.
\end{align}
Then, using Lemma~\ref{lem:Cauchy_Schwarz} with
\begin{align}
T=T_{i,n},
\qquad
\Pi=Q_n,
\qquad
c=e^{-n\delta},
\qquad
\rho=\overline{\rho}_i^{\otimes n},
\end{align}
we obtain
\begin{align}
\Tr[T_{i,n}\overline{\rho}_i^{\otimes n}]
\le
\left(
\sqrt{
e^{-n\delta}\Tr[\overline{\rho}_i^{\otimes n}Q_n]
}
+
\sqrt{
1-\Tr[\overline{\rho}_i^{\otimes n}Q_n]
}
\right)^2.
\end{align}
On the other hand, due to~\eqref{eq:T_n_convergence}, we have
\begin{align}
\lim_{n\to\infty}\Tr[T_{i,n}\overline{\rho}_i^{\otimes n}]= 1,
\end{align}
and hence,
\begin{align}
\lim_{n\to\infty}\left(
\sqrt{
e^{-n\delta}\Tr[Q_n\overline{\rho}_i^{\otimes n}]
}
+
\sqrt{
1-\Tr[Q_n\overline{\rho}_i^{\otimes n}]
}
\right)^2=1.
\end{align}
Therefore, we obtain
\begin{align}
    \lim_{n\to\infty}\Tr[(I_n-T_n)\overline{\rho}_i^{\otimes n}]=\lim_{n\to\infty}\Tr[Q_n\overline{\rho}_i^{\otimes n}]=0.
\end{align}
Consequently,
\begin{align}
    \label{eq:T_n_I}
\lim_{n\to\infty}\Tr\left[(I_n-T_n)\rho_n\right]
=
\sum_{i\in J} p_i \lim_{n\to\infty}\Tr\!\left[(I_n-T_n)\overline{\rho}_i^{\otimes n}\right]
=0.
\end{align}

Due to~\eqref{eq:T_n_II} and~\eqref{eq:T_n_I}, the test operator $T_n$ satisfies
\begin{align}
\lim_{n\to\infty}\Tr\!\left[(I_n-T_n)\rho_n\right]= 0,
\qquad
\liminf_{n\to\infty}
-\frac1n\log \Tr\!\left[T_n\sigma_n\right]\ge a+\delta.
\end{align}
By the hypothesis-testing characterization of
$\underline{D}(\vec{\rho}\|\vec{\sigma})$, equivalently by the
Neyman-Pearson optimality of the spectral tests, we have
\begin{align}
a\le \underline{D}(\vec{\rho}\|\vec {\sigma}).
\end{align}
Since $a$ in~\eqref{eq:a_condition} was arbitrary, we conclude that
\begin{align}
\underline{D}(\vec{\rho}\|\vec{\sigma})
\ge 
\min_i\left\{
\underline{D}(\vec{\rho}_i\|\vec{\sigma})
\right\}.
\end{align}
\end{proof}

The converse part of Proposition~\ref{key proposition to the main theorem} is shown as follows.

\begin{proposition}[Converse part]\label{main proposition: converse}
    Suppose that $\vec{\rho}$ is a sequence of mixed sources of $\vec{\rho}_i=\left\{\overline{\rho}_i^{\otimes n}\right\}_{n=1}^{\infty}$ defined on $\vec{\mathcal{H}}$ such that
    \begin{align}
        \rho_n = \sum_{i\in J} p_i \overline{\rho}_i^{\otimes n},
    \end{align}
    where $J$ is a finite set, $p_i>0$ for all $i\in J$, and $\sum_{i\in J} p_i = 1$.
    Let $\vec{\sigma}=\left\{\sigma_n\right\}_{n=1}^{\infty}$ be a sequence of density operators defined on $\vec{\mathcal{H}}$.
    Then, it holds that
    \begin{align}
        \underline{D}\left( \vec{\rho} \| \vec{\sigma} \right)
        \le
        \min_{i\in J}
        \left\{
            \underline{D}\left( \vec{\rho}_i \| \vec{\sigma} \right)
        \right\}.
    \end{align}
\end{proposition}

\begin{proof}
    Let $a$ be any real number satisfying
\begin{align}
    a<\underline{D}(\vec{\rho}\|\vec{\sigma}).
\end{align}
We write
\begin{align}
T_n\coloneqq \{\rho_n-e^{na}\sigma_n\geq 0\}.
\end{align}
By the definition of $\underline{D}(\vec{\rho}\|\vec{\sigma})$, i.e.,
\begin{align}
    \lim_{n\to\infty}\Tr[(I_n-T_n)\rho_n]=\sum_i p_i\lim_{n\to\infty}\Tr[(I_n-T_n)\overline{\rho}_i^{\otimes n}]=0,
\end{align}
the type-I error satisfies, for every $i$ such that $p(i)>0$,
\begin{align}
    \label{eq:assumption_T_n}
    \lim_{n\to\infty}\Tr[(I_n-T_n)\overline{\rho}_i^{\otimes n}]=0.
\end{align}
On the other hand, due to
\begin{align}
\Tr\!\left[T_n(\rho_n-e^{na}\sigma_n)\right]\ge 0,
\end{align}
the type-II error satisfies
\begin{align}
    \label{eq:T_n_sigma_n_bound}
\Tr\!\left[T_n \sigma_n\right]
\le e^{-na}\Tr\!\left[T_n \rho_n\right]
\le e^{-na}.
\end{align}

Fix $i\in J$ and $\delta>0$. Define
\begin{align}
T_{i,n}^{(\delta)}
\coloneqq\{\overline{\rho}_i^{\otimes n}-e^{n(a-\delta)}\sigma_n\geq 0\}.
\end{align}
For every operator $T$ with $0\le T\le I$, it holds that
\begin{align}
    \Tr[T_{i,n}^{(\delta)}(\overline{\rho}_i^{\otimes n}-e^{n(a-\delta)}\sigma_n)]\geq
    \Tr[T(\overline{\rho}_i^{\otimes n}-e^{n(a-\delta)}\sigma_n)],
\end{align}
and hence, we obtain the quantum Neyman--Pearson inequality for the pair
$(\overline{\rho}_i^{\otimes n},\sigma_n)$:
\begin{align}
\Tr\!\left[\left(I-T_{i,n}^{(\delta)}\right)\overline{\rho}_i^{\otimes n}\right]
+
e^{n(a-\delta)}
\Tr\!\left[T_{i,n}^{(\delta)}\sigma_n\right]
\le
\Tr\!\left[(I-T)\overline{\rho}_i^{\otimes n}\right]
+
e^{n(a-\delta)}
\Tr\!\left[T\sigma_n\right].
\end{align}
Applying this inequality to $T=T_n$, we have
\begin{align}
\Tr\!\left[\left(I-T_{i,n}^{(\delta)}\right)\overline{\rho}_i^{\otimes n}\right]
+
e^{n(a-\delta)}
\Tr\!\left[T_{i,n}^{(\delta)}\sigma_n\right]
\le
\Tr\!\left[\left(I_n-T_n\right)\overline{\rho}_i^{\otimes n}\right]
+
e^{n(a-\delta)}
\Tr\!\left[T_n\sigma_n\right].
\end{align}
Dropping the nonnegative second term on the left-hand side and using~\eqref{eq:T_n_sigma_n_bound},
we get
\begin{align}
\Tr\!\left[\left(I-T_{i,n}^{(\delta)}\right)\overline{\rho}_i^{\otimes n}\right]
\le
\Tr\!\left[\left(I_n-T_n\right)\overline{\rho}_i^{\otimes n}\right]
+
e^{-n\delta}.
\end{align}
Taking the limit and using~\eqref{eq:assumption_T_n} gives
\begin{align}
\lim_{n\to\infty}\Tr\!\left[\left(I-T_{i,n}^{(\delta)}\right)\overline{\rho}_i^{\otimes n}\right]
=0.
\end{align}

By the definition of $\underline{D}$, this implies
\begin{align}
a-\delta
\le
\underline{D}(\vec{\rho}_i\|\vec{\sigma}).
\end{align}
Since $\delta>0$ is arbitrary, we conclude that, for every $i$,
\begin{align}
a
\le
\underline{D}(\vec{\rho}_i\|\vec{\sigma})
.
\end{align}
Since $a$ can be taken arbitrarily close to $\underline{D}(\vec{\rho}\|\vec{\sigma})$, we have
\begin{align}
\underline{D}(\vec{\rho}\|\vec{\sigma})\le
\min_i\left\{
\underline{D}(\vec{\rho}_i\|\vec{\sigma})
\right\}.
\end{align}
\end{proof}

\subsection{The spectral inf-divergence rate in composite hypothesis testing}
\label{subsec: the spectral inf-divergence rate in composite hypothesis testing}

In the previous section, we have studied the spectral inf-divergence rate for mixed sources of quantum states in the first argument while maintaining the second argument as a general sequence of quantum states; as a corollary, when the second argument is a sequence of sets of states in composite hypothesis testing, we have the following characterization.
\begin{corollary}[Corollary of Proposition \ref{key proposition to the main theorem}]\label{characterization of underline D when the second argument is for sets}
    Let $S_n \subset \mathcal{D}(\mathcal{H}_n)$ be a subset of density operators on a finite-dimensional Hilbert space $\mathcal{H}_n$ for each $n\ge 1$,
    and let $\vec{S}$ be the sequence $\{S_n\}_{n=1}^\infty$.
    Also, suppose that $\vec{\rho}$ is a sequence of mixed sources defined on $\vec{\mathcal{H}}$ such that
    \begin{align}
        \rho_n = \sum_{i\in J} p_i \overline{\rho}_i^{\otimes n},
        \quad
        \vec{\rho}_i = \{\overline{\rho}_i^{\otimes n}\}_{n=1}^\infty,
    \end{align} 
    where $J$ is a finite set, $p_i>0$ for every $i\in J$, and $\sum_{i\in J} p_i=1$.
    Then, it holds that
    \begin{align}
        \inf_{\vec{\sigma}\in \vec{S}} \underline{D}(\vec{\rho}\|\vec{\sigma})
        = \min_i \inf_{\vec{\sigma}\in \vec{S}} \underline{D}(\vec{\rho}_i\|\vec{\sigma}).
    \end{align}
\end{corollary}
\begin{proof}
    The proof immediately follows from Proposition \ref{key proposition to the main theorem} since
    \begin{align}
        \inf_{\vec{\sigma}\in \vec{S}} \underline{D}(\vec{\rho}\|\vec{\sigma})
        &=
        \inf_{\vec{\sigma}\in \vec{S}} \min_i \underline{D}(\vec{\rho}_i\|\vec{\sigma})\\
        &=
        \min_i \inf_{\vec{\sigma}\in \vec{S}} \underline{D}(\vec{\rho}_i\|\vec{\sigma}).
    \end{align}
\end{proof}

\subsection{The relationship between two definitions of type-II error exponent}
\label{subsec: the relationship between two definitions of type-II error exponent}

In this subsection, we relate two ways of formulating the optimal type-II error exponent for composite hypothesis testing. The first formulation is the supremum exponent $B(\varepsilon|\vec{\rho}\|\vec{\sigma})$ in Definition~\ref{def:B} for a fixed sequence of alternative states $\vec{\sigma}$. The second formulation is the exponent obtained from the optimal type-II error $\beta_\varepsilon(\rho_n\|S_n)$ in Definition~\ref{def beta} against a set $S_n$ of alternative states. 

The following proposition shows that these two formulations agree up to an arbitrarily small relaxation of the type-I error threshold.
The proof is divided into three steps. First, Lemma~\ref{exchange} rewrites the composite type-II exponent in terms of fixed alternative sequences $\vec{\sigma}\in\vec{S}$. Then Proposition~\ref{prop:right inequality comparison} proves the right inequality in~\eqref{beta and B ineq}, which follows from relaxing the type-I constraint from $\varepsilon$ to $\varepsilon+\varepsilon_0$. Finally, Proposition~\ref{prop:left inequality comparison} proves the left inequality, which follows because any sequence of tests satisfying the blockwise constraint $\alpha_n\le\varepsilon$ is also admissible for the asymptotic constraint $\limsup_n\alpha_n\le\varepsilon$.
This comparison will be used in the proof of Theorem~\ref{Main theorem} to connect the information-spectrum characterization to the known generalized quantum Stein's lemma for IID states.

\begin{proposition}[Relationship between the two definitions of the optimal type-II exponent]\label{relationship between two definitions}
    For a sequence $\vec{\rho}=\left\{\rho_n\right\}_{n=1}^{\infty}$ of density operators defined on $\vec{\mathcal{H}}=\left\{\mathcal{H}_n\right\}_{n=1}^{\infty}$ and a sequence $\vec{S}=\left\{S_n\subset \mathcal{D}(\mathcal{H}_n)\right\}_{n=1}^{\infty}$ of sets of states satisfying
    \textnormal{(\ref{set1})}, \textnormal{(\ref{set2})}, \textnormal{(\ref{set3})}, and \textnormal{(\ref{set4})},
    it holds that for
    any $\varepsilon\in [0,1)$ and $\varepsilon_0 \in (0, 1-\varepsilon]$, 
    \begin{align}\label{beta and B ineq}
        \liminf_{n\to\infty} -\frac{1}{n}\log \beta_{\varepsilon}\left(\rho_n \middle\| S_n\right)
        \le \inf_{\vec{\sigma}\in \vec{S}} B(\varepsilon| \vec{\rho}\|\vec{\sigma})
        \le \liminf_{n\to\infty} -\frac{1}{n}\log \beta_{\varepsilon+\varepsilon_0}\left(\rho_n \middle\| S_n\right).
    \end{align}
    Moreover, if $\varepsilon \in (0, 1)$ and $\vec{\rho}$ is a sequence of IID states $\left\{\rho^{\otimes n}\right\}_{n=1}^{\infty}$, then
    \begin{align}
        \inf_{\vec{\sigma}\in \vec{S}} B(\varepsilon| \vec{\rho}\|\vec{\sigma})
        =
        \liminf_{n\to\infty} -\frac{1}{n}\log \beta_{\varepsilon'}\left(\rho^{\otimes n} \middle\| S_n\right)
    \end{align}
    for any $\varepsilon' \in (0, 1)$.
\end{proposition}

\begin{proof}

The first inequality in~\eqref{beta and B ineq} follows from Proposition~\ref{prop:left inequality comparison} and Lemma~\ref{exchange}. The second inequality in~\eqref{beta and B ineq} follows from Proposition~\ref{prop:right inequality comparison} and Lemma~\ref{exchange}. If $\varepsilon\in(0,1)$ and $\vec{\rho}=\{\rho^{\otimes n}\}_{n=1}^{\infty}$ is IID, then Theorem~\ref{generalized quantum Stein's lemma} implies that the leftmost and rightmost terms in~\eqref{beta and B ineq} are equal for any type-I error threshold in $(0,1)$. Hence the middle term is also equal to the same value, proving the last assertion.
\end{proof}

We first record the minimax and $\liminf$ exchange step used in both directions of the comparison.

\begin{lemma}[exchange of $\liminf$, $\min$, $\max$]\label{exchange}
    For a sequence $\vec{\rho}=\left\{\rho_n\right\}_{n=1}^{\infty}$ of density operators defined on $\vec{\mathcal{H}}=\left\{\mathcal{H}_n\right\}_{n=1}^{\infty}$ and a sequence $\vec{S}=\left\{S_n\subset \mathcal{D}(\mathcal{H}_n)\right\}_{n=1}^{\infty}$ of 
    sets of states satisfying
    \textnormal{(\ref{set1})}, \textnormal{(\ref{set2})}, \textnormal{(\ref{set3})}, and \textnormal{(\ref{set4})},
    it holds that
    \begin{align}
        \liminf_{n\to\infty} -\frac{1}{n} \log \beta_{\varepsilon}\left(\rho_n \middle\| S_n\right)
        =
        \inf_{\vec{\sigma}\in\vec{S}} \liminf_{n\to\infty} \max_{T_n}\left\{
            -\frac{1}{n}\log \Tr[T_n \sigma_n] \middle| 0 \le T_n \le I, \Tr[(I_n-T_n)\rho_n] \le \varepsilon
        \right\}
    \end{align}
\end{lemma}

\begin{proof}
    The proof follows from the following chain of equalities:
    \begin{align}
        &\liminf_{n\to\infty} -\frac{1}{n} \log \beta_{\varepsilon}\left(\rho_n \middle\| S_n\right)\\
        &=
        \liminf_{n\to\infty} -\frac{1}{n} \log \min \left\{
            \max_{\sigma_n\in S_n} \Tr[T_n\sigma_n] \middle| 0 \le T_n \le I, \Tr[(I_n-T_n)\rho_n] \le \varepsilon
        \right\}\\
        &=\label{2nd eq}
        \liminf_{n\to\infty} -\frac{1}{n}\log \max_{\sigma_n\in S_n} \min \left\{
            \Tr[T_n\sigma_n] \middle| 0 \le T_n \le I, \Tr[(I_n-T_n)\rho_n] \le \varepsilon
        \right\}\\
        &=
        \liminf_{n\to\infty} \min_{\sigma_n\in S_n} -\frac{1}{n}\log \min \left\{
            \Tr[T_n\sigma_n] \middle| 0 \le T_n \le I, \Tr[(I_n-T_n)\rho_n] \le \varepsilon
        \right\}\\
        &=\label{last eq}
        \inf_{\vec{\sigma}\in\vec{S}} \liminf_{n\to\infty} \max_{T_n}\left\{
            -\frac{1}{n}\log \Tr[T_n \sigma_n] \middle| 0 \le T_n \le I, \Tr[(I_n-T_n)\rho_n] \le \varepsilon
        \right\},
    \end{align}
    where~\eqref{2nd eq} follows from the minimax theorem~\cite{hayashi2025generalizedquantumsteinslemma,vonNeumann1928,sion1958,komiya1988} due to compactness and convexity of relevant sets and bilinearity of the objective function, and~\eqref{last eq} follows from
\begin{align}
    \liminf_{n\to\infty}\inf_{x_n\in X_n} f_n(x_n)
    =
    \inf_{x\in X}\liminf_{n\to\infty} f_n(x_n)
\end{align}
for all sequences $X=\{X_n\}_{n=1}^\infty$ of sets.
\end{proof}

We next prove the right inequality in~\eqref{beta and B ineq}. The only point is that a sequence of tests satisfying an asymptotic type-I constraint $\limsup_n\alpha_n\le\varepsilon$ satisfies the slightly relaxed blockwise constraint $\alpha_n\le\varepsilon+\varepsilon_0$ for all sufficiently large $n$.

\begin{proposition}[Right inequality in Proposition~\ref{relationship between two definitions}]

\label{prop:right inequality comparison}

Under the assumptions of Proposition~\ref{relationship between two definitions}, for any $\varepsilon\in[0,1)$ and $\varepsilon_0\in(0,1-\varepsilon]$, it holds that

\begin{align}
\inf_{\vec{\sigma}\in \vec{S}} B(\varepsilon| \vec{\rho}\|\vec{\sigma})
\le
\liminf_{n\to\infty} -\frac{1}{n}\log \beta_{\varepsilon+\varepsilon_0}\left(\rho_n \middle\| S_n\right).
\end{align}
\end{proposition}

\begin{proof}
    Due to lemma \ref{exchange}, it suffices to show that for any $\vec{\sigma}\in\vec{S}$,
    \begin{align}\label{right ineq: main}
        B(\varepsilon| \vec{\rho} \| \vec{\sigma})
        \le
        \liminf_{n\to\infty} \max_{T_n}\left\{
            -\frac{1}{n}\log \Tr[T_n \sigma_n] \middle| 0 \le T_n \le I, \Tr[(I_n-T_n)\rho_n] \le \varepsilon+\varepsilon_0
        \right\}.
    \end{align}
    Choose an arbitrary sequence of test operators $\vec{T}=\{T_n\}_{n=1}^{\infty}$ satisfying
    $\limsup_{n\to\infty} \alpha_n(T_n) \le \varepsilon$.
    Then, there exists $N\in\mathbb{N}$ such that for every $n\ge N$,
    \begin{align}
        \alpha_n(T_n) \le \varepsilon+\varepsilon_0.
    \end{align}
    Next, prepare a new sequence $\vec{T}'=\{T_n'\}_{n=1}^{\infty}$ defined as
    \begin{align}
        T_n' = \begin{cases}
            T_n & \text{if } n\ge N, \\
            1 & \text{if } n < N.
        \end{cases}
    \end{align}
    For this $\vec{T}'$, we have $\alpha_n(T_n') \le \varepsilon+\varepsilon_0$ for every $n$, and
    \begin{align}
        \liminf_{n\to\infty} -\frac{1}{n}\log \Tr[T_n' \sigma_n]
        =
        \liminf_{n\to\infty} -\frac{1}{n}\log \Tr[T_n \sigma_n].
    \end{align}
    Thus, it follows that
    \begin{align}\label{right ineq: part 1}
        &\sup_{\vec{T}=\{T_n\}_{n=1}^{\infty}}\left\{
            \liminf_{n\to\infty} -\frac{1}{n}\log \Tr[T_n \sigma_n] \middle| 0 \le T_n \le I, \limsup_{n\to\infty} \alpha_n(T_n) \le \varepsilon
        \right\}
        \\&\le
        \sup_{\vec{T}=\{T_n\}_{n=1}^{\infty}}\left\{
            \liminf_{n\to\infty} -\frac{1}{n}\log \Tr[T_n \sigma_n] \middle| 0 \le T_n \le I, \forall n. \left[\Tr[(I_n-T_n)\rho_n] \le \varepsilon+\varepsilon_0\right]
        \right\}.
    \end{align}
    Choose an arbitrary sequence $\vec{T}=\{T_n\}_{n=1}^{\infty}$ satisfying $0 \le T_n \le I$ and $\Tr[(I_n-T_n)\rho_n] \le \varepsilon+\varepsilon_0$ for every $n$.
    Then for this sequence, it follows that
    \begin{align}
        -\frac{1}{n}\log \Tr[T_n \sigma_n]
        \le
        \max_{T_n}\left\{
            -\frac{1}{n}\log \Tr[T_n \sigma_n] \middle| 0 \le T_n \le I, \Tr[(I_n-T_n)\rho_n] \le \varepsilon+\varepsilon_0
        \right\}.
    \end{align}
    Taking $\liminf$ on both sides gives
    \begin{align}
        \liminf_{n\to\infty} -\frac{1}{n}\log \Tr[T_n \sigma_n]
        \le
        \liminf_{n\to\infty} \max_{T_n}\left\{
            -\frac{1}{n}\log \Tr[T_n \sigma_n] \middle| 0 \le T_n \le I, \Tr[(I_n-T_n)\rho_n] \le \varepsilon+\varepsilon_0
        \right\}.
    \end{align}
    Since this holds for arbitrary $\vec{T}$, we have
    \begin{align}\label{right ineq: part 2}
        &\sup_{\vec{T}=\{T_n\}_{n=1}^{\infty}}\left\{
            \liminf_{n\to\infty} -\frac{1}{n}\log \Tr[T_n \sigma_n] \middle| 0 \le T_n \le I, \forall n. \left[\Tr[(I_n-T_n)\rho_n] \le \varepsilon+\varepsilon_0\right]
        \right\}
        \\&\le
        \liminf_{n\to\infty} \max_{T_n}\left\{
            -\frac{1}{n}\log \Tr[T_n \sigma_n] \middle| 0 \le T_n \le I, \Tr[(I_n-T_n)\rho_n] \le \varepsilon+\varepsilon_0
        \right\}.
    \end{align}
    Combining~\eqref{right ineq: part 1} and~\eqref{right ineq: part 2}, ~\eqref{right ineq: main} is proved.
\end{proof}

We finally prove the left inequality in~\eqref{beta and B ineq}. Here no relaxation of the type-I error threshold is needed, because a test satisfying the constraint at every blocklength automatically satisfies the asymptotic constraint in the definition of $B(\varepsilon|\vec{\rho}\|\vec{\sigma})$.

\begin{proposition}[Left inequality in Proposition~\ref{relationship between two definitions}]
\label{prop:left inequality comparison}
Under the assumptions of Proposition~\ref{relationship between two definitions}, for any $\varepsilon\in[0,1)$, it holds that
\begin{align}
\liminf_{n\to\infty} -\frac{1}{n}\log \beta_{\varepsilon}\left(\rho_n \middle\| S_n\right)
\le
\inf_{\vec{\sigma}\in \vec{S}} B(\varepsilon| \vec{\rho}\|\vec{\sigma}).
\end{align}
\end{proposition}

\begin{proof}
    Due to lemma \ref{exchange}, it suffices to show that for any $\vec{\sigma}\in\vec{S}$,
    \begin{align}\label{left ineq: main}
        \liminf_{n\to\infty} \max_{T_n}\left\{
            -\frac{1}{n}\log \Tr[T_n \sigma_n] \middle| 0 \le T_n \le I, \Tr[(I_n-T_n)\rho_n] \le \varepsilon
        \right\}
        \le
        B(\varepsilon| \vec{\rho} \| \vec{\sigma}).
    \end{align}
    Let $\vec{T}^*=\{T_n^*\}_{n=1}^{\infty}$ be the optimal sequence for the left-hand side.
    Then the left-hand side of~\eqref{left ineq: main} is equal to
    \begin{align}
        \liminf_{n\to\infty} -\frac{1}{n}\log \Tr[T_n^* \sigma_n],\label{optimal sequence}
    \end{align}
    satisfying $\alpha_n(T_n^*)\le \varepsilon$.
    For this sequence, it holds that
    \begin{align}
        \limsup_{n\to\infty} \alpha_n(T_n^*) \le \varepsilon,
    \end{align}
    thus (\ref{optimal sequence}) is an element of the set
    \begin{align}
        \left\{
            \liminf_{n\to\infty} -\frac{1}{n}\log \Tr[T_n \sigma_n] \middle| 0 \le T_n \le I, \limsup_{n\to\infty} \alpha_n(T_n) \le \varepsilon
        \right\}.
    \end{align}
    By the definition of $B(\varepsilon| \vec{\rho} \| \vec{\sigma})$, we have
    \begin{align}
        \liminf_{n\to\infty} -\frac{1}{n}\log \Tr[T_n^* \sigma_n]
        \le
        B(\varepsilon| \vec{\rho} \| \vec{\sigma}).
    \end{align}
\end{proof}

\subsection{Proof of Theorem~\ref{Main theorem}}
\label{subsec: proof of the main theorem}

We now prove Theorem~\ref{Main theorem} by combining the results established in the previous subsections. The key input is Corollary~\ref{characterization of underline D when the second argument is for sets}, which reduces the spectral inf-divergence rate of the mixed source to the worst IID component. We then use the information-spectrum identity between $B$ and $\underline{D}$, and finally apply Proposition~\ref{relationship between two definitions} together with the generalized quantum Stein's lemma for IID states, Theorem~\ref{generalized quantum Stein's lemma}, to identify the resulting exponent with the regularized relative entropy.

\begin{proof}[Proof of Theorem~\ref{Main theorem}]
Using the information-spectrum characterization
$B(\vec{\rho}\|\vec{\sigma})
=
\underline{D}(\vec{\rho}\|\vec{\sigma})
$
in Lemma~\ref{lem:B=D}, together with Corollary~\ref{characterization of underline D when the second argument is for sets}, we obtain
\begin{align}
\inf_{\vec{\sigma}\in\vec{S}} B(\vec{\rho}\|\vec{\sigma})
&=
\inf_{\vec{\sigma}\in\vec{S}} \underline{D}(\vec{\rho}\|\vec{\sigma})\\
&=
\min_{i\in J}
\inf_{\vec{\sigma}\in\vec{S}}
\underline{D}(\vec{\rho}_i\|\vec{\sigma})\\
&=
\min_{i\in J}
\inf_{\vec{\sigma}\in\vec{S}}
B(\vec{\rho}_i\|\vec{\sigma}).
\label{eq:main_proof_reduction_to_components}
\end{align}
Since $B(\varepsilon|\vec{\rho}_i\|\vec{\sigma})$ is right-continuous at $\varepsilon=0$~\cite[Eq.~(20)]{4069150}, by the definition of $B(\vec{\rho}_i\|\vec{\sigma})$ in Definition~\ref{def:B}, we have
\begin{align}
B(\vec{\rho}_i\|\vec{\sigma})
=
\inf_{\varepsilon>0}
B(\varepsilon|\vec{\rho}_i\|\vec{\sigma}).
\end{align}
Therefore,
\begin{align}
\inf_{\vec{\sigma}\in\vec{S}} B(\vec{\rho}\|\vec{\sigma})
&=
\min_{i\in J}
\inf_{\vec{\sigma}\in\vec{S}}
\inf_{\varepsilon>0}
B(\varepsilon|\vec{\rho}_i\|\vec{\sigma})\\
&=
\min_{i\in J}
\inf_{\varepsilon>0}
\inf_{\vec{\sigma}\in\vec{S}}
B(\varepsilon|\vec{\rho}_i\|\vec{\sigma}).
\label{eq:main_proof_epsilon_reduction}
\end{align}
For each fixed IID component $\vec{\rho}_i=\{\overline{\rho}_i^{\otimes n}\}_{n=1}^{\infty}$ and each $\varepsilon\in(0,1)$, Proposition~\ref{relationship between two definitions} shows that
\begin{align}
\inf_{\vec{\sigma}\in\vec{S}}
B(\varepsilon|\vec{\rho}_i\|\vec{\sigma})=
        \liminf_{n\to\infty} -\frac{1}{n}\log \beta_{\varepsilon'}\left(\overline{\rho}_i^{\otimes n} \middle\| S_n\right)
\end{align}
for any $\varepsilon'\in(0,1)$.
By Theorem~\ref{generalized quantum Stein's lemma}, this exponent is given by the regularized relative entropy:
\begin{align}
\liminf_{n\to\infty} -\frac{1}{n}\log \beta_{\varepsilon'}\left(\overline{\rho}_i^{\otimes n} \middle\| S_n\right)
=
\lim_{n\to\infty}
\frac{1}{n}
\min_{\sigma_n\in S_n}
D(\overline{\rho}_i^{\otimes n}\|\sigma_n).
\end{align}
The right-hand side is independent of $\varepsilon,\varepsilon'\in(0,1)$. Hence,
\begin{align}
\inf_{\vec{\sigma}\in\vec{S}} B(\vec{\rho}\|\vec{\sigma})
&=
\min_{i\in J}
\inf_{\varepsilon>0}
\lim_{n\to\infty}
\frac{1}{n}
\min_{\sigma_n\in S_n}
D(\overline{\rho}_i^{\otimes n}\|\sigma_n)\\
&=
\min_{i\in J}
\lim_{n\to\infty}
\frac{1}{n}
\min_{\sigma_n\in S_n}
D(\overline{\rho}_i^{\otimes n}\|\sigma_n).
\end{align}
\end{proof}

\section{Counterexample to analogy in $\varepsilon$-hypothesis testing for mixed sources}
\label{sec:counterexample_epsilon}

In Theorem~\ref{Main theorem} and Proposition~\ref{key proposition to the main theorem}, we have characterized the exponent in composite hypothesis testing for mixed sources in the case where the type-I error is required to vanish asymptotically. 
One may naturally expect that an analogous characterization continues to hold for any fixed $\varepsilon\in(0,1)$, as in the generalized quantum Stein's lemma for IID states in the first argument in Theorem~\ref{generalized quantum Stein's lemma}. 
In this section, we show that this expectation is not correct for mixed sources in the first argument.
In Sec.~\ref{subsec: counterexample}, we present the main theorem in this section, which is the characterization of the supremum type-II error exponent in $\varepsilon$-hypothesis testing of a mixed source against an IID source. In Sec.~\ref{subsec: Analysis of the spectral inf-divergence rate using pinching map}, we show the spectral inf-divergence rate does not change when we apply the pinching map~\cite{Hayashi_2002} in terms of $\vec{\sigma}$ to $\vec{\rho}$, and thus the classical results can be directly applied. In Sec.~\ref{proof of the sub main theorem}, we prove the main theorem of this section by using the classical results. 

\subsection{Counterexample}
\label{subsec: counterexample}

Our purpose here is not to establish a composite hypothesis-testing theorem for arbitrary alternatives under a nonzero type-I error constraint; rather, we exhibit the obstruction already in a counterexample, where the second argument is IID quantum states. 
This restriction is also natural from the viewpoint of the information-spectrum method: in the classical theory, a clean formula for $\varepsilon$-hypothesis testing of mixed sources is obtained when the alternative is IID sources. 
We prove the quantum analogue of this formula and use it to show that, for $\varepsilon\in(0,1)$, the exponent is no longer determined solely by the worst component in the mixture. 
Instead, it depends on how the allowed type-I error $\varepsilon$ can discard components of the mixed source with small relative entropy.

We first state a characterization of the type-II error exponent in $\varepsilon$-hypothesis testing for mixed sources against an IID source, and so that we can use it to present a counterexample to a direct extension of Theorem~\ref{Main theorem} to $\varepsilon\in(0,1)$.
The theorem is the quantum extension of the classical characterization of $\varepsilon$-hypothesis testing for mixed sources in the information-spectrum method~\cite{han2003}.

\begin{theorem}[Characterization of the type-II error exponent in $\varepsilon$-hypothesis testing for a mixed source against an IID source]\label{Sub main theorem}
    Let $\vec{\rho}=\left\{\rho_n\right\}_{n=1}^{\infty}$ be a sequence of mixed sources defined on $\vec{\mathcal{H}}$ such that
    \begin{align}
        \rho_n = \sum_{i\in J} p_i \overline{\rho}_i^{\otimes n},
    \end{align}
    where $J$ is a finite set, $p_i>0$ for all $i\in J$, and $\sum_{i\in J} p_i = 1$, and
    let $\vec{\sigma}=\left\{\sigma^{\otimes n}\right\}_{n=1}^{\infty}$ be a sequence of IID states defined on $\vec{\mathcal{H}}$.
    We denote $d_i=D(\overline{\rho}_i\|\sigma)$.
    Then, it holds that for any $\varepsilon\in[0, 1)$,
    \begin{align}
        B(\varepsilon| \vec{\rho} \| \vec{\sigma}) = \sup\left\{ R \middle| \sum_{i: d_i \le R} p_i \le \varepsilon \right\}.
    \end{align}
\end{theorem}

The theorem shows that the achievable exponent is determined by the largest threshold $R$ for which the total weight of components whose relative entropy is at most $R$ does not exceed the allowed type-I error $\varepsilon$. In particular, the exponent can jump as $\varepsilon$ crosses the cumulative weights of the low-relative-entropy components, as shown below.

\begin{example}[Counterexample to the worst-component characterization for $\varepsilon>0$]\label{counterexample varepsilon}
    Suppose that $\vec{\rho}=\left\{\rho_n\right\}_{n=1}^{\infty}$ is a sequence of mixed sources of two IID states defined on $\vec{\mathcal{H}}$, i.e.,
    \begin{align}
        \rho_n = p \overline{\rho}_1^{\otimes n} + (1-p) \overline{\rho}_2^{\otimes n},
    \end{align}
    where $p\in(0,1)$, and $\vec{\sigma}=\left\{\sigma^{\otimes n}\right\}_{n=1}^{\infty}$ is a sequence of IID states defined on $\vec{\mathcal{H}}$.
    Also, suppose that $D(\overline{\rho}_1||\sigma) < D(\overline{\rho}_2||\sigma)$.
    Then, it holds that for any $\varepsilon\in[0, 1)$,
    \begin{align}
        B(\varepsilon| \vec{\rho} \| \vec{\sigma}) = \begin{cases}
            D(\overline{\rho}_1||\sigma) & \text{if } 0 \le \varepsilon < p, \\
            D(\overline{\rho}_2||\sigma) & \text{if } p \le \varepsilon < 1.
        \end{cases}
    \end{align}
    Theorem~\ref{Main theorem} shows that, when $\varepsilon=0$, the supremum type-II error exponent in composite hypothesis testing for mixed sources of quantum states is determined by the worst component in the mixture.
    However, the present example shows that this worst-component characterization does not generally extend to $\varepsilon$-hypothesis testing with $\varepsilon\in(0,1)$, even when the second argument is restricted to an IID source.
    Indeed, once the allowed type-I error satisfies $\varepsilon\ge p$, the component $\overline{\rho}_1^{\otimes n}$ of weight $p$ can be discarded, and the exponent jumps from $D(\overline{\rho}_1\|\sigma)$ to $D(\overline{\rho}_2\|\sigma)$.
\end{example}

\subsection{Analysis of the spectral inf-divergence rate using pinching map}
\label{subsec: Analysis of the spectral inf-divergence rate using pinching map}

We analyze whether an analogy of Proposition~\ref{key proposition to the main theorem} with $\varepsilon=0$ holds for $\varepsilon\in[0, 1]$.
For this purpose, we show a pinching reduction for the spectral inf-divergence rate.
One direction follows from the monotonicity of hypothesis testing under completely positive and trace-preserving (CPTP) maps, while the converse direction uses the pinching inequality and requires that the number of distinct eigenvalues of the second argument grows subexponentially.
The pinching is defined as follows.

\begin{definition}[Pinching maps with respect to a sequence of quantum states]
    Let $\vec{\rho}=\{\rho_n\}_{n=1}^{\infty}$ and $\vec{\sigma}=\{\sigma_n\}_{n=1}^{\infty}$ be sequences of density operators defined on $\vec{\mathcal{H}}=\{\mathcal{H}_n\}_{n=1}^{\infty}$.
    Let the spectral decomposition of $\sigma_n$ be
    \begin{align}
        \sigma_n = \sum_{i=1}^{d_n} \lambda_{n, i}\Pi_{n, i},
    \end{align}
    where $d_n$ is the number of distinct eigenvalues of $\sigma_n$, and $\Pi_{n, i}$ is the corresponding projection for every $i$.
    For each $n$, we define a pinching map $\mathcal{E}_n: \mathcal{D}(\mathcal{H}_n)\to \mathcal{D}(\mathcal{H}_n)$ in terms of $\sigma_n$ as
    \begin{align}
        \mathcal{E}_n(\rho) = \sum_{i=1}^{d_n} \Pi_{n, i}\rho\Pi_{n, i}.
    \end{align}
    We also define the pinched sequence
    \begin{align}
        \mathcal{E}(\vec{\rho}) = \{\mathcal{E}_n(\rho_n)\}_{n=1}^{\infty}.
    \end{align}
By definition, $\mathcal{E}_n(\rho_n)$ commutes with $\sigma_n$ for every $n$, and hence we can utilize any classical results for the pair $\mathcal{E}(\vec{\rho})$ and $\vec{\sigma}$.
\end{definition}

Using the pinching map, the main proposition in this section is stated as follows.
We may obtain Proposition \ref{key proposition to the main theorem} as a special case (with $\varepsilon=0$) but with the additional assumption that $d_n=e^{o(n)}$.

\begin{proposition}[Pinching reduction for the spectral inf-divergence rate]
\label{sub main proposition}
    Let $\vec{\rho}=\{\rho_n\}_{n=1}^{\infty}$ and $\vec{\sigma}=\{\sigma_n\}_{n=1}^{\infty}$ be sequences of density operators,
    and suppose that $d_n=e^{o(n)}$, i.e., the number of distinct eigenvalues of $\sigma_n$ grows subexponentially in $n$.
    Then, for any $\varepsilon\in[0, 1]$, it holds that 
    \begin{align}
        \underline{D}(\varepsilon|\vec{\rho}\|\vec{\sigma})
        =
        \underline{D}(\varepsilon|\mathcal{E}(\vec{\rho})\|\vec{\sigma}).
    \end{align}
\end{proposition}

\begin{proof}
    The inequality
    \begin{align}
        \underline{D}(\varepsilon|\vec{\rho}\|\vec{\sigma})
        \le
        \underline{D}(\varepsilon|\mathcal{E}(\vec{\rho})\|\vec{\sigma})
    \end{align}
    is shown by Proposition~\ref{re: direct part of pinching}, which uses the pinching inequality and the assumption $d_n=e^{o(n)}$.
    Conversely, we show Proposition~\ref{prop:CPTP_monotonicity_spectral_inf_divergence} and apply it to the pinching map $\mathcal{E}_n$ to obtain
    \begin{align}
        \underline{D}(\varepsilon|\vec{\rho}\|\vec{\sigma})
        \ge
        \underline{D}(\varepsilon|\mathcal{E}(\vec{\rho})\|\mathcal{E}(\vec{\sigma})).
    \end{align}
    Since $\mathcal{E}_n(\sigma_n)=\sigma_n$, we obtain
    \begin{align}
        \underline{D}(\varepsilon|\vec{\rho}\|\vec{\sigma})
        \ge
        \underline{D}(\varepsilon|\mathcal{E}(\vec{\rho})\|\vec{\sigma}).
    \end{align}
    Combining these inequalities proves the claim.
\end{proof}

For the first direction of the inequality, we prepare the following lemma.
\begin{lemma}[Cauchy-Schwarz estimate for two projections]\label{Cauchy-Schwarz 2}
    Let $\rho$ be a density operator and $\Pi_1$, $\Pi_2$ be projections.
    Then it holds that
    \begin{align}
        \Tr[\Pi_1\rho]
        \le
        \left(
            \sqrt{\Tr[\Pi_2\rho]}
            +
            \sqrt{\Tr[\Pi_1(I-\Pi_2)\rho(I-\Pi_2)]}
        \right)^2.
    \end{align}
\end{lemma}
\begin{proof}
    Let $\Pi_2' = I - \Pi_2$. Then, we decompose
    \begin{align}
        \Tr[\Pi_1\rho]
        &=
        \Tr[\Pi_1\Pi_2\rho\Pi_2]
        +
        \Tr[\Pi_1\Pi_2'\rho\Pi_2']
        +
        \Tr[\Pi_1\Pi_2\rho\Pi_2']
        +
        \Tr[\Pi_1\Pi_2'\rho\Pi_2].
    \end{align}
    Since $\Pi_2\rho\Pi_2$ is positive semidefinite, we have
    \begin{align}
        \Tr[\Pi_1\Pi_2\rho\Pi_2]
        \le
        \Tr[\Pi_2\rho\Pi_2]
        =
        \Tr[\Pi_2\rho].
    \end{align}
    By Cauchy-Schwarz inequality,
    \begin{align}
        \left|\Tr[\Pi_1\Pi_2\rho\Pi_2']\right|
        &=
        \left|\Tr[\Pi_1\Pi_2\rho\Pi_2'\Pi_1]\right|\\
        &=
        \left| \Tr[(\sqrt{\rho}\Pi_2\Pi_1)^\dagger (\sqrt{\rho}\Pi_2'\Pi_1)] \right|\\
        &\le
        \sqrt{
            \Tr\left[ \Pi_1\Pi_2\rho\Pi_2\Pi_1 \right]
            \Tr\left[ \Pi_1\Pi_2'\rho\Pi_2'\Pi_1 \right]
        }\\
        &\le
        \sqrt{
            \Tr\left[
                \Pi_2\rho
            \right]
            \Tr\left[
                \Pi_1\Pi_2'\rho\Pi_2'
            \right]
        }.
    \end{align}
    We have the same bound for $\Tr[\Pi_1\Pi_2'\rho\Pi_2]$.
    Therefore, we obtain the desired bound.
\end{proof}

Using this, we prove the following inequality.

\begin{proposition}[Pinching reduction]\label{re: direct part of pinching}
Let $\vec{\rho}=\{\rho_n\}_{n=1}^{\infty}$ and $\vec{\sigma}=\{\sigma_n\}_{n=1}^{\infty}$ be sequences of density operators,
and let $\mathcal{E}(\vec{\rho})$ be the pinched sequence defined in terms of $\vec{\sigma}$.
Suppose $d_n=e^{o(n)}$. Then, for any $\varepsilon\in[0,1]$, it holds that
\begin{align}
\underline{D}(\varepsilon|\vec{\rho}\|\vec{\sigma})
\le
\underline{D}(\varepsilon|\mathcal{E}(\vec{\rho})\|\vec{\sigma}).
\end{align}
\end{proposition}

\begin{proof}
    Let $a$ be any real number satisfying $a<\underline{D}(\varepsilon|\vec{\rho}\|\vec{\sigma})$.
    It suffices to show that $a\le \underline{D}(\varepsilon|\mathcal{E}(\vec{\rho})\|\vec{\sigma})$.
    Let $\delta>0$ be arbitrary. Define
    \begin{align}
        S_n&\coloneqq \{\rho_n-e^{na}\sigma_n\geq 0\},\\
        T_n&\coloneqq \{\mathcal{E}_n(\rho_n)-e^{n(a-\delta)}\sigma_n\geq 0\}.
    \end{align}
    Use lemma (\ref{Cauchy-Schwarz 2}) with $\Pi_1=S_n$, $\Pi_2=T_n$, and $\rho=\rho_n$ to obtain
    \begin{align}\label{use case of Cauchy-Schwarz 2}
        \Tr[S_n\rho_n]
        \le
        \left(
            \sqrt{\Tr[T_n\rho_n]}
            +
            \sqrt{\Tr[S_n(I_n-T_n)\rho_n(I_n-T_n)]}
        \right)^2.
    \end{align}

    By the pinching inequality \cite[Lemma 3.10]{hayashi2017qit}, we have
    \begin{align}
        \rho_n \le d_n \mathcal{E}_n(\rho_n),
    \end{align}
    and thus
    \begin{align}
        (I_n-T_n)\rho_n(I_n-T_n)
        \le
        d_n (I_n-T_n)\mathcal{E}_n(\rho_n)(I_n-T_n).
    \end{align}
    Since $S_n$ is positive semidefinite, we have
    \begin{align}\label{ineq 1}
        \Tr\left[ S_n(I_n-T_n)\rho_n(I_n-T_n) \right]
        &\le
        d_n \Tr\left[ S_n(I_n-T_n)\mathcal{E}_n(\rho_n)(I_n-T_n) \right].
    \end{align}

    By the definition of $T_n$, we have
    \begin{align}
        (I_n-T_n)\mathcal{E}_n(\rho_n)(I_n-T_n)
        \le
        e^{n(a-\delta)}(I_n-T_n)\sigma_n(I_n-T_n)
        \le
        e^{n(a-\delta)}\sigma_n,
    \end{align}
    where the second inequality holds since $\sigma_n$ commutes with $(I_n-T_n)$. Thus, we have
    \begin{align}\label{ineq 2}
        \Tr\left[ S_n(I_n-T_n)\mathcal{E}_n(\rho_n)(I_n-T_n) \right]
        &\le
        e^{n(a-\delta)} \Tr\left[ S_n\sigma_n \right].
    \end{align}

    Likewise, by the definition of $S_n$, we have
    \begin{align}
        S_n \sigma_n S_n
        \le
        e^{-na} S_n \rho_n S_n,
    \end{align}
    and thus
    \begin{align}\label{ineq 3}
        \Tr\left[ S_n\sigma_n \right]
        &\le
        e^{-na} \Tr\left[ S_n\rho_n \right]\\
        &\le e^{-na}.
    \end{align}

    Combining~\eqref{ineq 1},~\eqref{ineq 2}, and~\eqref{ineq 3}, we obtain
    \begin{align}
        \Tr\left[ S_n(I_n-T_n)\rho_n(I_n-T_n) \right]
        &\le
        d_n \Tr\left[ S_n(I_n-T_n)\mathcal{E}_n(\rho_n)(I_n-T_n) \right]\\
        &\le
        d_n e^{n(a-\delta)} \Tr\left[ S_n\sigma_n \right]\\
        &\le
        d_n e^{-n\delta} \Tr\left[ S_n\rho_n \right]\\
        &\le
        d_n e^{-n\delta}.
    \end{align}
    This inequality and (\ref{use case of Cauchy-Schwarz 2}) leads to
    \begin{align}
        \Tr[S_n\rho_n]
        &\le
        \left(
            \sqrt{\Tr[T_n\rho_n]}
            +
            \sqrt{d_n e^{-n\delta}}
        \right)^2\\
        &\le
        \Tr[T_n\rho_n]
        +
        2\sqrt{d_n e^{-n\delta}}
        +
        d_n e^{-n\delta}.
    \end{align}
    Since $d_n=e^{o(n)}$, it follows $d_n e^{-n\delta} \to 0$ as $n \to \infty$.
    Therefore, we have
    \begin{align}
        \liminf_{n\to\infty} \Tr[S_n\rho_n]
        &\le
        \liminf_{n\to\infty} \Tr[T_n\rho_n]
        =
        \liminf_{n\to\infty} \Tr[T_n\mathcal{E}_n(\rho_n)].
    \end{align}
    Equivalently,
    \begin{align}
        \limsup_{n\to\infty} \Tr\left[ (I_n-T_n)\mathcal{E}_n(\rho_n) \right]
        \le
        \limsup_{n\to\infty} \Tr\left[ (I_n-S_n)\rho_n \right].
    \end{align}
    Since $a<\underline{D}(\varepsilon|\vec{\rho}\|\vec{\sigma})$, the right-hand side is at most $\varepsilon$, and thus
    \begin{align}
        \limsup_{n\to\infty} \Tr\left[ (I_n-T_n)\mathcal{E}_n(\rho_n) \right]
        \le
        \varepsilon.
    \end{align}
    This implies that
    \begin{align}
        a-\delta
        \le
        \underline{D}(\varepsilon|\mathcal{E}(\vec{\rho})\|\vec{\sigma}).
    \end{align}
    Since $\delta>0$ was arbitrary, we conclude that
    \begin{align}
        a\le
        \underline{D}(\varepsilon|\mathcal{E}(\vec{\rho})\|\vec{\sigma}),
    \end{align}
    which completes the proof.
\end{proof}

The reverse direction follows from the following monotonicity property of the spectral inf-divergence rate under CPTP maps.

\begin{proposition}[Monotonicity of the spectral inf-divergence rate]
\label{prop:CPTP_monotonicity_spectral_inf_divergence}
    Let $\vec{\rho}=\{\rho_n\}_{n=1}^{\infty}$ and $\vec{\sigma}=\{\sigma_n\}_{n=1}^{\infty}$ be sequences of density operators defined on $\vec{\mathcal{H}}=\{\mathcal{H}_n\}_{n=1}^{\infty}$.
    Let $\mathcal{E}_n'$ be a completely positive and trace-preserving (CPTP) map acting on $\mathcal{D}(\mathcal{H}_n)$, and define $\mathcal{E}'(\vec{\rho})=\{\rho_n'=\mathcal{E}_n'(\rho_n)\}_{n=1}^\infty$ and $\mathcal{E}'(\vec{\sigma})=\{\sigma_n'=\mathcal{E}_n'(\sigma_n)\}_{n=1}^\infty$.
    Then, for any $\varepsilon\in[0,1]$, it holds that
    \begin{align}
        \underline{D}(\varepsilon | \vec{\rho}\|\vec{\sigma}) \ge \underline{D}(\varepsilon | \mathcal{E}'(\vec{\rho}) \| \mathcal{E}'(\vec{\sigma})).
    \end{align}
\end{proposition}
\begin{proof}
    Using \eqref{B=D}, it suffices to show the following inequality:
    \begin{align}
        B(\varepsilon | \vec{\rho}\|\vec{\sigma}) \ge B(\varepsilon | \mathcal{E}'(\vec{\rho}) \| \mathcal{E}'(\vec{\sigma})).
    \end{align}
    For a sequence of test $\vec{T}'=\{T_n'\}_{n=1}^\infty$, we have
    \begin{align}
        \Tr\left[\rho_n'T_n'\right]
        &= \Tr\left[\mathcal{E}_n'(\rho_n)T_n'\right]\\
        &= \Tr\left[\rho_n \mathcal{E}_n^{\prime\dagger}(T_n')\right].
    \end{align}
    where $\mathcal{E}_n^{\prime\dagger}$ is the adjoint of $\mathcal{E}_n'$, and $T_n=\mathcal{E}_n'^\dagger(T_n')$ is a valid test since $\mathcal{E}_n'^\dagger$ is unital and completely positive.
    Similarly, we have
    \begin{align}
        \Tr\left[\sigma_n'T_n'\right] = \Tr\left[\sigma_n T_n\right].
    \end{align}
    Thus, it follows that
    \begin{align}
        &\left\{ \liminf_{n\to\infty} -\frac{1}{n}\log \Tr[T_n' \sigma_n'] \middle| 0 \le T_n' \le I_n, \limsup_{n\to\infty} \Tr[(I_n-T_n')\rho_n'] \le \varepsilon \right\}
        \\&\subset
        \left\{\liminf_{n\to\infty} -\frac{1}{n}\log \Tr[T_n \sigma_n] \middle| 0 \le T_n \le I_n, \limsup_{n\to\infty} \Tr[(I_n-T_n)\rho_n] \le \varepsilon \right\}.
    \end{align}
    Taking $\sup$ over $T_n'$ and $T_n$ yields the desired inequality.
\end{proof}

\subsection{Proof of Theorem \ref{Sub main theorem}}
\label{proof of the sub main theorem}

Now we are ready to prove Theorem~\ref{Sub main theorem}. Proposition~\ref{sub main proposition} states that the spectral inf-divergence rate does not change when we apply, to $\rho_n$, the pinching map with respect to $\sigma_n$. As $\mathcal{E}_n(\rho_n)$ commutes with $\sigma_n$, we can directly apply the classical results.

\begin{proof}[Proof of Theorem~\ref{Sub main theorem}]
    Let
    \begin{align}
        \rho_n
        =
        \sum_{i\in J}
        p_i\overline{\rho}_i^{\otimes n},
    \end{align}
    where $J$ is finite, $p_i>0$ for all $i\in J$, and $\sum_{i\in J}p_i=1$.
    Let
    \begin{align}
        \vec{\sigma}
        =
        \{\sigma^{\otimes n}\}_{n=1}^{\infty}
    \end{align}
    be an IID source.

    Since $\sigma$ acts on a fixed finite-dimensional Hilbert space, the number of distinct eigenvalues of $\sigma^{\otimes n}$ grows at most polynomially in $n$.
    In particular, if $d_n$ denotes the number of distinct eigenvalues of $\sigma^{\otimes n}$, then
    \begin{align}
        d_n
        =
        e^{o(n)}.
    \end{align}
    Hence Proposition~\ref{sub main proposition} gives
    \begin{align}
        \underline{D}(\varepsilon|\vec{\rho}\|\vec{\sigma})
        =
        \underline{D}(\varepsilon|\mathcal{E}(\vec{\rho})\|\vec{\sigma}).
        \label{eq:pinching_reduction_mixed_source}
    \end{align}

    For every $n$, the states $\mathcal{E}_n(\rho_n)$ and $\sigma^{\otimes n}$ commute.
    Therefore, after simultaneous diagonalization, the pair $\mathcal{E}(\vec{\rho})$ and $\vec{\sigma}$ can be regarded as a pair of classical sources.
    Moreover,
    \begin{align}
        \mathcal{E}_n(\rho_n)
        =
        \sum_{i\in J}
        p_i\mathcal{E}_n(\overline{\rho}_i^{\otimes n}),
    \end{align}
    so $\mathcal{E}(\vec{\rho})$ is a finite mixture of the classical sources obtained from
    \begin{align}
        \mathcal{E}(\vec{\rho}_i)
        \coloneqq
        \{\mathcal{E}_n(\overline{\rho}_i^{\otimes n})\}_{n=1}^{\infty}.
    \end{align}

    For each $i\in J$, define
    \begin{align}
        d_i
        \coloneqq
        D(\overline{\rho}_i\|\sigma).
    \end{align}
    By Proposition~\ref{sub main proposition} applied to each IID component $\vec{\rho}_i=\{\overline{\rho}_i^{\otimes n}\}_{n=1}^{\infty}$, together with the quantum Stein lemma for IID states, we have
    \begin{align}
        \underline{D}(\varepsilon|\mathcal{E}(\vec{\rho}_i)\|\vec{\sigma})
        =
        \underline{D}(\varepsilon|\vec{\rho}_i\|\vec{\sigma})
        =
        d_i
    \end{align}
    for every $i\in J$ and every $\varepsilon\in[0,1)$.
    Hence the information spectrum of the $i$-th component concentrates at the single point $d_i$.

    We now apply the classical information-spectrum formula for $\varepsilon$-hypothesis testing~\cite[Theorem~4.2.1]{han2003}.
    For the classical pair corresponding to $\mathcal{E}(\vec{\rho})$ and $\vec{\sigma}$, let $P_{X^n}$ denote the probability distribution corresponding to $\mathcal{E}_n(\rho_n)$, and let $P_{Y^n}$ denote the probability distribution corresponding to $\sigma^{\otimes n}$.
    Define
    \begin{align}
        K(R)
        :=
        \limsup_{n\to\infty}
        \Pr
        \left[
            \frac{1}{n}
            \log
            \frac{
                P_{X^n}(X^n)
            }{
                P_{Y^n}(X^n)
            }
            \le
            R
        \right],
    \end{align}
    where $X^n$ is distributed according to $P_{X^n}$.
    Then Ref.~\cite{chen} (see also Ref.~\cite[Theorem~4.2.1]{han2003}) gives
    \begin{align}
        B(\varepsilon|\mathcal{E}(\vec{\rho})\|\vec{\sigma})
        =
        \sup
        \left\{
            R
            :
            K(R)
            \le
            \varepsilon
        \right\}.
        \label{eq:han_information_spectrum_formula}
    \end{align}

    It remains to evaluate $K(R)$ for the present finite mixed source.
    We write $P_{X_i^n}$ for the probability distribution corresponding to $\mathcal{E}_n(\overline{\rho}_i^{\otimes n})$.
    Then
    \begin{align}
        P_{X^n}
        =
        \sum_{i\in J}
        p_iP_{X_i^n},
    \end{align}
    so the random variable $X^n$ can be generated by first choosing an index $I\in J$ according to the distribution $\{p_i\}_{i\in J}$ and then drawing $X_i^n$ according to $P_{X_i^n}$.
    Therefore,
    \begin{align}
        K(R)
        =
        \limsup_{n\to\infty}
        \sum_{i\in J}
        p_i
        \Pr
        \left[
            \frac{1}{n}
            \log
            \frac{
                P_{X^n}(X_i^n)
            }{
                P_{Y^n}(X_i^n)
            }
            \le
            R
        \right],
        \label{eq:K_decomposition_by_components}
    \end{align}
    where $X_i^n$ is distributed according to $P_{X_i^n}$.
    Note that the numerator in~\eqref{eq:K_decomposition_by_components} is the mixed distribution $P_{X^n}$, not the component distribution $P_{X_i^n}$.

    We now compare the likelihood ratio of the mixed distribution with that of each component.
    Since
    \begin{align}
        P_{X^n}(x^n)
        \ge
        p_iP_{X_i^n}(x^n)
    \end{align}
    for all $i\in J$ and all $x^n$, we have
    \begin{align}
        \frac{1}{n}
        \log
        \frac{
            P_{X^n}(x^n)
        }{
            P_{Y^n}(x^n)
        }
        \ge
        \frac{1}{n}
        \log
        \frac{
            P_{X_i^n}(x^n)
        }{
            P_{Y^n}(x^n)
        }
        +
        \frac{1}{n}\log p_i.
        \label{eq:mixture_lower_bound_component_likelihood}
    \end{align}
    Conversely, the argument on the information spectrum for mixed sources in Example~4.2.1 of~\cite{han2003} shows that the additional contributions from the other components do not change the asymptotic normalized log-likelihood ratio associated with the $i$-th component at the exponential scale.
    Consequently, the $i$-th component contributes weight $p_i$ to the information-spectrum distribution at the point $d_i$.

    For each $i\in J$, the component spectrum concentrates at the single point
    $d_i = D(\overline{\rho}_i\|\sigma)$,
    as shown above from Proposition~\ref{sub main proposition} and the quantum Stein lemma.
    Therefore, at every continuity point of the finite step function
    \begin{align}
        R
        \mapsto
        \sum_{i\in J:d_i\le R}
        p_i,
    \end{align}
    the distribution function $K(R)$ satisfies
    \begin{align}
        K(R)
        =
        \sum_{i\in J:d_i<R}
        p_i.
        \label{eq:K_finite_mixed_source}
    \end{align}
    The distinction between $d_i<R$ and $d_i\le R$ only concerns the finitely many discontinuity points $R=d_i$.
    This boundary ambiguity does not affect the supremum in~\eqref{eq:han_information_spectrum_formula}.
    Consequently,
    \begin{align}
        \sup
        \left\{
            R
            :
            K(R)
            \le
            \varepsilon
        \right\}
        =
        \sup
        \left\{
            R
            :
            \sum_{i\in J:d_i\le R}
            p_i
            \le
            \varepsilon
        \right\}.
        \label{eq:finite_mixture_spectrum_formula}
    \end{align}

    Combining~\eqref{eq:han_information_spectrum_formula} and~\eqref{eq:finite_mixture_spectrum_formula}, we obtain
    \begin{align}
        B(\varepsilon|\mathcal{E}(\vec{\rho})\|\vec{\sigma})
        =
        \sup
        \left\{
            R
            :
            \sum_{i\in J:d_i\le R}
            p_i
            \le
            \varepsilon
        \right\}.
        \label{eq:epsilon_exponent_pinched_mixed_source}
    \end{align}
    Finally, Lemma~\ref{lem:B=D}, together with~\eqref{eq:pinching_reduction_mixed_source} and~\eqref{eq:epsilon_exponent_pinched_mixed_source}, yields
    \begin{align}
        B(\varepsilon|\vec{\rho}\|\vec{\sigma})
        =
        \sup
        \left\{
            R
            :
            \sum_{i\in J:d_i\le R}
            p_i
            \le
            \varepsilon
        \right\}.
    \end{align}
\end{proof}

\section{Conclusion}
\label{sec:conclusion}

In this work, we prove a theorem that generalizes the generalized quantum Stein's lemma under the axioms of Ref.~\cite{hayashi2025generalizedquantumsteinslemma} to the setting where the null hypothesis is given by a mixture of IID states rather than a single IID state.
The theorem shows that, in the limit of vanishing type-I error $\varepsilon=0$, the infimum of the supremum type-II error exponent is determined by the regularized relative entropy of resource of the least distinguishable component in the mixture.
Our proof of the main theorem, Theorem~\ref{Main theorem}, uses the information spectrum method~\cite{han2003,4069150} to analyze mixed sources of quantum states.
We develop quantum-specific techniques to overcome the non-commutativity of quantum states in this analysis; nevertheless, the resulting characterization of the supremum type-II error exponent is analogous to the classical formula for mixed sources~\cite{han2003}.

Moreover, we show that an analogous worst-case-component characterization does not generally hold for $\varepsilon$-hypothesis testing with $\varepsilon>0$, by presenting a counterexample in which the supremum $\varepsilon$-achievable type-II error exponent is not determined by the worst-case component, in contrast to the case $\varepsilon=0$.
To analyze this, we use the pinching technique, which has been previously applied to quantum hypothesis testing for IID null hypotheses~\cite{hiai1991proper,Hayashi_2002,hayashi2025generalizedquantumsteinslemma}, but was not used in the information-spectrum analysis of general non-IID quantum sources in Ref.~\cite{4069150}.
Under the additional assumption that the number of distinct eigenvalues of the alternative hypothesis $\sigma_n$ grows subexponentially in $n$, we showed that this pinching reduction preserves the relevant spectral inf-divergence rate and reduces the problem to a commuting, hence classical, setting.
We can then apply the classical information-spectrum formula for mixed sources~\cite{han2003} to obtain the characterization of $\varepsilon$-hypothesis testing.

These quantum results are closely aligned with the classical case.
For $\varepsilon=0$, the supremum type-II error exponent admits a concise characterization as the minimum of the relative entropies of the components; by contrast, for $\varepsilon>0$, the corresponding exponent for mixed sources against general sources does not admit such a concise characterization~\cite{han2003}.
It is worth noting that our analysis uses different techniques for the vanishing-error case $\varepsilon=0$ and for the finite-error case $\varepsilon\in[0,1)$. For $\varepsilon=0$, we characterize the spectral inf-divergence rate of a mixed source as the minimum of the rates of its components by extending the information-spectrum-based argument to handle non-commutativity directly, without any additional assumption. By contrast, the case $\varepsilon>0$ is analyzed under the additional assumption that $\sigma_n$ has at most subexponentially many distinct eigenvalues, i.e., $d_n=e^{o(n)}$, which allows us to reduce the problem to the classical case. These two approaches are complementary, and we expect that the techniques developed here will be useful for future research in both directions.

Finally, our results suggest several directions for further development.
One natural direction is to extend the present analysis beyond mixed sources, including quantum extensions of known results in the classical information spectrum method~\cite{han2003}.
It would also be interesting to investigate applications of our results to quantum resource theories, in the spirit of the generalized quantum Stein's lemma for quantum resource theories~\cite{Brand_o_2008,brandao2010reversible,brandao2010generalization,Brandao2015,hayashi2025generalizedquantumsteinslemma,hayashi2025generalizedquantumsteinslemmaCQ,bergh2025generalizedquantumsteinslemma}.

\textit{Note added.---}
We mention the independent concurrent work by Lami et al.~\cite{lami2026}.
Our main result, Theorem~\ref{Main theorem}, yields, due to Proposition~\ref{relationship between two definitions},
\begin{align}\label{note added: ours}
\lim_{\varepsilon\to +0} \liminf_{n\to\infty} -\frac{1}{n}\log \beta_\varepsilon(\rho_n\|S_n)
=
\min_{\rho \in \mathcal{R}}\lim_{n\to\infty}\frac{1}{n}\min_{\sigma_n\in S_n}D(\rho^{\otimes n}\|\sigma_n),
\end{align}
where $\mathcal{R}=\{\overline{\rho}_i:i\in J\}$ is a finite subset of $\mathcal{D}(\mathcal{H})$, $\{\rho_n=\sum_{i\in J}p_i\overline{\rho}_i^{\otimes n}\}_{n=1}^{\infty}$ is a mixed source of the IID states in $\mathcal{R}$, and $p_i$ is an arbitrary probability distribution.
By contrast, their main result~\cite[Theorem 1]{lami2026} states that, for any $\varepsilon\in (0, 1)$,
\begin{align}\label{note added: theirs}
    \lim_{n\to\infty} \inf_{\rho_n\in \operatorname{conv}(\mathcal{R}_n^{\prime\mathrm{iid}})}
    -\frac{1}{n}\log\beta_\varepsilon(\rho_n\|S_n)
    =
    \inf_{\rho\in\mathcal{R}^\prime}\lim_{n\to\infty}\frac{1}{n}\min_{\sigma_n\in S_n} D(\rho^{\otimes n}\|\sigma_n),
\end{align}
where $\mathcal{R}'$ is a subset of $\mathcal{D}(\mathcal{H})$, and $\operatorname{conv}(\mathcal{R}_n^{\prime\mathrm{iid}})$ is the convex hull of $\mathcal{R}_n^{\prime\mathrm{iid}} \coloneqq \{\rho^{\otimes n} : \rho \in \mathcal{R}^\prime\}$.
In their analysis, $S_n$ is assumed to satisfy the Brand\~{a}o--Plenio axioms~\cite{brandao2010generalization} for each $n$.
Apparently, these results look similar, but there are several differences that are worth noting.

First, our result holds for any fixed mixed source with an arbitrary probability distribution $p_i$, while their result characterizes the type-II error exponent after an optimization over the convex hull of IID states.
In our setting with arbitrary $p_i$, the strong converse does not generally hold, unlike in their setting; indeed, we have given a counterexample (Example~\ref{counterexample varepsilon}) where the optimal type-II error exponent under a fixed nonzero type-I error threshold changes discontinuously when the probability distribution defining the mixture is changed.

Second, in our result, the number $|\mathcal{R}|$ of components of the mixed source is assumed to be finite; if $|\mathcal{R}|$ is infinite, then it is larger than any exponential function of $n$, and the result cannot hold in general, as shown by the counterexample in Example~\ref{counterexample if J_n grows}.
By contrast, in their work, the set $\mathcal{R}'$ of IID states is subject to no cardinality restriction, while the type-II error exponent is characterized after the optimization over its convex hull.

Finally, the Brand\~{a}o–Plenio axioms are more restrictive than our assumptions: the Brandão–Plenio axioms require, in addition to our assumptions, closedness under tracing out subsystems and closedness under permutations of subsystems.
It remains unclear whether these additional axioms are necessary for \eqref{note added: theirs} to hold.

Therefore, these two results, \eqref{note added: ours} and \eqref{note added: theirs}, are closely related yet not directly comparable.
It would be interesting to further clarify the relationship between them and to seek a unifying result that encompasses both \eqref{note added: ours} and \eqref{note added: theirs} as special cases.

\begin{acknowledgments}
    HY thanks Ludovico Lami for discussions.
    This work was supported by the Google UTokyo Quantum Partnership Research Program, Faculty Research Funding from Google Quantum AI\@, JST PRESTO Grant Number JPMJPR201A, JPMJPR23FC, JSPS KAKENHI Grant Number JP23K19970, JST CREST Grant Number JPMJCR25I5, and JST [Moonshot R\&D] [Grant Number JPMJMS256J]\@.
\end{acknowledgments}

\bibliography{citation}
\end{document}